\documentclass[transmag]{IEEEtran}

\usepackage[T1]{fontenc}
\usepackage{amsthm} 
\usepackage[cmintegrals]{newtxmath}
\usepackage{cite}
\usepackage{psfrag}
\usepackage{color,comment}
\usepackage{gensymb}
\usepackage{makecell}
\usepackage[utf8]{inputenc}
\usepackage[english]{babel}
\usepackage[]{verbatim}
\usepackage{syntonly}
\usepackage{textcomp,xcolor}
\usepackage{amsmath,balance}
\usepackage{upgreek}
\usepackage{graphicx}

\usepackage{varioref}
\usepackage{hyperref}
\usepackage[noabbrev]{cleveref}

\newtheorem{Proposition}{Proposition}
\newtheorem{Corollary }{Corollary}

\def\BibTeX{{\rm B\kern-.05em{\sc i\kern-.025em b}\kern-.08em T\kern-.1667em\lower.7ex\hbox{E}\kern-.125emX}}
\begin{document}

\title{Link Budget Analysis for Reconfigurable Smart Surfaces in Aerial Platforms}

\author{Safwan Alfattani, Wael Jaafar, Yassine Hmamouche, Halim Yanikomeroglu, and Abbas Yongaçoglu 
\thanks{This work is supported in part by a
scholarship from King AbdulAziz University, Saudi Arabia, in part by
the National Sciences and Engineering Research Council of Canada
(NSERC), and in part by Huawei Canada.}
\thanks{S. Alfattani is with the Electrical Engineering Department of King Abdulaziz University, Saudi Arabia, and with the
	School of Electrical Engineering and Computer Science,
	University of Ottawa, Ottawa, ON, Canada, (e-mail: smalfattani@kau.edu.sa).}
\thanks{W. Jaafar and H. Yanikomeroglu are with the Department of Systems and
Computer Engineering, Carleton University, Ottawa, ON, Canada, e-mails:
{(waeljaafar, halim}@sce.carleton.ca)}
\thanks{Y. Hmamouche is with the Mathematical and Electrical Engineering Department, IMT-Atlantique, Brest, France,  (e-mail: yassine.hmamouche@imt-atlantique.fr).}
\thanks{A. Yongacoglu is with the
 School of Electrical Engineering and Computer Science, University of Ottawa, Ottawa, ON, Canada 
 (e-mail: yongac@uottawa.ca).}}

\IEEEtitleabstractindextext{\begin{abstract}
Non-terrestrial networks, including Unmanned Aerial Vehicles (UAVs), High Altitude Platform Station (HAPS) \textcolor{black}{nodes} and \textcolor{black}{Low Earth Orbiting (LEO) satellites,} are expected to have a pivotal role 
\textcolor{black}{in sixth-generation wireless networks.}
With 
inherent features such as flexible placement, wide footprints, and preferred channel conditions, they can tackle several challenges 
\textcolor{black}{faced by current terrestrial networks.}
However, their successful and widespread adoption relies on 
energy-efficient on-board communication systems.
In this context, the integration of Reconfigurable Smart Surfaces (RSS) into aerial platforms is envisioned as a key enabler of energy-efficient and cost-effective 
\textcolor{black}{aerial platform deployments.}
 RSS consist of low-cost reflectors capable of smartly directing signals in a nearly passive way. 
\textcolor{black}{In this paper, we investigate} the link budget of RSS-assisted communications 
\textcolor{black}{for two RSS reflection paradigms discussed in the literature, namely ``specular'' and ``scattering'' paradigms.}
 Specifically, we analyze the characteristics of RSS-equipped aerial platforms and compare their communication performance with that of RSS-assisted terrestrial networks using standardized channel models. In addition, we derive the 
 \textcolor{black}{optimal aerial platform placements for both} reflection paradigms.
\textcolor{black}{Our} results provide important insights for the design of RSS-assisted communications. For instance, given that a HAPS has a large 
\textcolor{black}{area for RSS,} it provides superior link budget performance 
\textcolor{black}{in most studied scenarios.} 
In contrast, the limited RSS area on UAVs and the large propagation loss in \textcolor{black}{LEO satellite} communications make them unfavorable candidates for supporting terrestrial users. Finally, the optimal location of an RSS-equipped platform may depend on the platform's altitude, coverage footprint, and type of environment.

\end{abstract}


}

\maketitle

\section{Introduction}
As the fifth generation (5G) of wireless systems \textcolor{black}{are} being actively deployed, researchers in the wireless community \textcolor{black}{have} started investigating new technologies and innovative solutions to tackle the challenges and fulfill the demands of 
\textcolor{black}{of next-generation (6G) networks.}
One of the main 
\textcolor{black}{challenges involves supporting ubiquitous connectivity with high data rates in an energy efficient way.} 
With the inherent limitations of terrestrial environments, non-terrestrial networks are 
envisioned as an enabling technology for ubiquitous connectivity in future wireless communications. 
\textcolor{black}{Non-terrestrial networks including such platforms  as  Unmanned Aerial Vehicles (UAVs), High Altitude Platform Stations (HAPS) nodes, and Low Earth Orbit (LEO) satellites are capable of addressing such challenges  as coverage holes, blind spots, sudden increases in throughput demand, and terrestrial network failures. They can address these challenges  due to their  wider coverage footprints, strong line-of-sight (LoS) links, and flexibility of deployment compared to terrestrial
networks \cite{alfattani2020aerial,kurt2020vision,rinaldi2020non,kodheli2020satellite,zeng2019accessing}.} Moreover, the standardization efforts of the  Third Generation Partnership Project (3GPP) aiming to utilize aerial platforms for 5G 
\textcolor{black}{and beyond have made significant progress, as}
demonstrated by the standardization documents TR 38.811 \cite{3gpp2017Technical}, TR 22.829 \cite{3gpp2017Technical_2}, and TS 22.125 \cite{3gpp2017Technical_3}. Furthermore, several commercial projects are either in their initial phases of deployment or under development, which aim to design different types of aerial platforms capable of supporting wireless communications. Such projects include the Starlink LEO constellation by SpaceX \cite{starlink}, the Stratobus HAPS by Thales \cite{thales}, and the Nokia Drone Networks \cite{Nokia}.
Nevertheless, aerial platforms 
\textcolor{black}{are not yet a cutting-edge technology,}
and their current size, weight, and power (SWAP) limitations need to be further improved. 

On the other hand, 
reconfigurable smart surfaces (RSS) 
\textcolor{black}{have recently been introduced as an energy-efficient enabling technology}
for next-generation wireless networks \cite{di2019smart}.\textcolor{black}{\footnote{\textcolor{black}{RSS are referred to in the literature by other names, such as software-controlled metasurfaces}
\cite{liaskos2018new}, reconfigurable intelligent surfaces (RIS) \cite{Basar2019},  intelligent reflecting surfaces (IRS) \cite{Wu2019} \textcolor{black}{and} smart reflect-arrays \cite{Tan2018}.}} An RSS is a thin, lightweight metasurface
integrated with passive electronic components or switches to provide  unique and controlled manipulation of the wireless signals. It can alter the amplitude of the impinging signal, adjust its phase, and direct it to a target in a nearly passive way \cite{liaskos2018new,Basar2019}.
The deployment and utilization of RSS in terrestrial networks has been extensively studied, and 
\textcolor{black}{several research works, prototypes, \textcolor{black}{and} industrial experiments, summarized in \cite{wu2020intelligent, gong2020toward}, demonstrated the potential of this technology.}


Given the potential spectral and energy efficiencies of RSS-assisted communications and the stringent energy requirements of communications through  aerial platforms, equipping the latter with RSS 
\textcolor{black}{presents} an attractive solution to the SWAP issue. Indeed, due to the low-cost and negligible energy consumption of RSS reflectors, their use in aerial platforms is expected to support low-cost wireless communications for an extended flight duration. 
\textcolor{black}{In our previous work }\cite{alfattani2020aerial}, we discussed the feasibility of integrating RSS 
\textcolor{black}{in aerial platforms of different types.}
We 
\textcolor{black}{proposed a control architecture, detailed potential use cases, and examined associated challenges.}
In the context of UAVs only, 
the authors of \cite{lu2020enabling} showed that using RSS in UAVs enables a panoramic view of the environment, which can provide full-angle $360\degree$ signal reflections compared to $180\degree$ reflections in RSS-assisted terrestrial networks.
In \cite{abdalla2020uavs}, the authors presented potential use cases of RSS mounted 
\textcolor{black}{on UAVs and discussed}  related challenges and research opportunities. Similarly, \cite{shang2021uav} investigated the potential of RSS-equipped UAV swarms, where a use case was studied to demonstrate the achievable data rate performance of such systems.
The authors of \cite{samir2020optimizing} 
\textcolor{black}{studied the problem of wireless sensor data collection, where sensors were assisted by an RSS mounted on a UAV to reach the collecting sink.}
The objective was to maintain data freshness through accurate optimization of the UAV's location and the RSS phase-shifting configuration. Finally, in the context of LEO satellites,
the authors of \cite{tekb2020reconfigurable} investigated the utilization of RSS to support inter-satellite links in the  terahertz (THz) band. The results demonstrated a significant performance improvement in terms of bit error rate compared to non-RSS-assisted communications.

\textcolor{black}{Previous works have not thoroughly investigated RSS-enabled communication links, and so a link budget analysis for RSS-assisted non-terrestrial networks remains unexamined.}
\textcolor{black}{In contrast, a number of works have studied path-loss models for RSS-assisted terrestrial communications}
\cite{Basar2019,Tang2019,yildirim2020modeling,Ozdogan2019,nadeem2019intelligent}. While most of these models are  based on mathematical analysis using different approaches, some \textcolor{black}{of them have been} experimentally validated \cite{Tang2019}. 
These studies revealed the existence of two regimes that 
govern the performance of RSS-assisted communication systems. The first  is the ``specular'' reflection paradigm, where the path-loss model is analyzed using geometrical optics and imaging theory\textcolor{black}{. The} second is the ``scattering'' reflection paradigm, which obeys 
plate scattering theory and radar cross-section analysis. 
The factors that determine the governing regime of the RSS-assisted systems are the geometrical size of the RSS units, the communication frequency, and the distances separating an RSS from the transmitter and receiver.
Typically, 
\textcolor{black}{when an RSS is within a relatively short distance from a transmitter and/or a}  
\textcolor{black}{receiver, or when the RSS units}
are electrically 
\textcolor{black}{large (e.g., their dimensions are ten times larger than the wavelength denoted by $\lambda$), the path loss is governed by the specular reflection paradigm}
\cite{Basar2019,Tang2019,Ntontin2019a}. 
\textcolor{black}{Otherwise, the RSS-assisted communication follows the plate scattering reflection paradigm (i.e., when the distances between the RSS units and transmitter or receiver are large or when the RSS unit dimensions are very small)}
\cite{basar2020simris,ellingson2019path}. 
\textcolor{black}{It should be noted that}
the scattering paradigm can be designated as \textcolor{black}{``far-field'' paradigm,} whereas the specular reflection 
\textcolor{black}{can be referred to as ``near-field'' paradigm.}

Due to the specific design and environmental characteristics of aerial platforms compared to terrestrial systems, 
\textcolor{black}{the former, when equipped with an RSS, are expected to have a different link budget analysis.}
\textcolor{black}{Therefore, it is necessary to investigate and assess the benefits of RSS-enabled aerial platforms.}
Three major factors impact the feasibility of RSS-enabled aerial platforms, namely the operating frequency or wavelength, the platform’s surface area reserved for RSS units, and the operating altitude. 
While higher frequency signals are preferable for larger capacity links and enabling the deployment of more RSS reflector units, 
\textcolor{black}{higher frequency signals}
are more vulnerable to path-loss degradation 
from the communication distance and atmospheric attenuation. Also, larger platform RSS sizes may lead to a higher reflection gain, which may not be realizable for 
practical platform sizes.
Finally, 
although platforms operating at higher altitudes might be preferable due to their larger coverage footprint, they suffer from excessive propagation losses that may not be compensated \textcolor{black}{for even 
large area of RSS.}


In this paper, we aim to provide the link budget analysis for RSS-enabled aerial platform communication systems 
\textcolor{black}{for specular and scattering reflection}
paradigms. The received power of the system, for different RSS-enabled platforms, is calculated while taking into account the signal strength losses due to the specific characteristics of each platform. To minimize signal loss, we derive the optimal platform location and 
\textcolor{black}{maximum feasible number}
of RSS reflectors on-board  each platform. Link budget expressions are then derived for realistic communication conditions, as 
\textcolor{black}{defined by 3GPP standards.}
Finally, numerical results are provided to support the proposed link budget analysis. The contributions of the paper are highlighted as follows: 
\begin{enumerate}
    \item To the best of our knowledge, the performance parameters of RSS-enabled communications have been derived only for terrestrial networks and not for non-terrestrial systems, where different signal losses are experienced due to specific characteristics of aerial platforms and atmospheric phenomena. The impact of these factors is taken into account in this work.
    
    \item 
    We investigate the link budget analysis for RSS-enabled aerial platform communications 
    \textcolor{black}{for different platform types,} namely UAVs, HAPS nodes, and LEO satellites. 
    \textcolor{black}{To improve performance, we optimize platform locations and the maximum feasible number of mounted RSS reflectors.}
    
    \item We extend the link budget analysis to more realistic channel conditions, as defined in the 3GPP standards. 
    \textcolor{black}{We also provide numerical results to support the related parameters evaluation and the link budget analysis.
 }
\end{enumerate}

The rest of the paper is organized as follows. Section \ref{sec:budget_terr} discusses the conditions of the reflection and scattering paradigms, and then analyzes the link budget for RSS-assisted terrestrial networks.  Section \ref{sec:budget_non_terr} then exposes the characteristics of aerial platforms and derives the optimal platform placement for RSS-assisted aerial 
\textcolor{black}{communications. This section also investigates the related link budget for both reflection paradigms.}
Numerical results for the terrestrial and non-terrestrial RSS-assisted systems are presented and elaborated in Section \ref{sec:results}. Finally, Section \ref{sec:conc} 
concludes the paper.

\section{Link Budget Analysis for RSS-Assisted Terrestrial Networks}\label{sec:budget_terr}

\textcolor{black}{Here,} we present the link budget analysis for RSS-enabled terrestrial communications 
\textcolor{black}{(e.g., RSS units mounted on a building).}
Given the identified specular and scattering reflection paradigms, we derive the related received power expressions. 

Typically, terrestrial environments are characterized by 
blockages that result in high path loss, especially in dense-urban and urban environments. 
Accordingly, terrestrial network planning \textcolor{black}{depends} on cellular densification, where multiple base stations (BSs) are deployed in a relatively small area to ensure coverage of all users within the area. 
\textcolor{black}{But this comes at the expense of additional costs and inter-cell interference.}
To alleviate such inconveniences, the RSS can be deployed on the facades of buildings and used to either extend the cellular coverage or improve the signal quality in poorly served areas. As shown in Fig. \ref{Fig:terrestrial_RSS}, 
\textcolor{black}{the signal forwarded by the RSS from the BS to the user equipment (UE) can either}
substitute the direct link when the latter is absent, or it can  be added constructively  to the weak direct link in order to strengthen the received signal.

\begin{figure}[t]
	\centering
	\includegraphics[width=\linewidth]{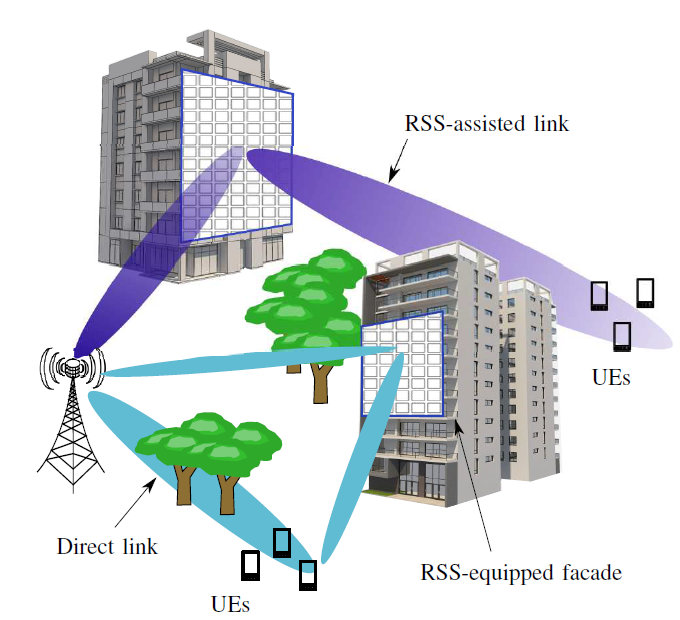}
	\caption{System model of RSS in a terrestrial environment. }
	\label{Fig:terrestrial_RSS}
\end{figure}

\subsection{The Specular Reflection Paradigm}
\textcolor{black}{The relation that governs the paradigm}
of the RSS-assisted communications has been defined in \cite{Tang2019,Gab2020} as follows:
\begin{equation}\label{Eq:dis_lim}
d_{\rm lim} = \dfrac{2 A_t}{\lambda},
\end{equation}
\textcolor{black}{where $d_{\rm lim}$   denotes the maximum distance between
the RSS and either the transmitter (Tx) or the receiver (Rx)
in the specular reflection paradigm\footnote{An example of (Tx,Rx) in the terrestrial environment is (BS,UE).}, and  $A_t$ is the total RSS area.}

When the length and width dimensions of the RSS units are large enough (i.e., above $10 \lambda$), and the distance separating the RSS from the Tx/Rx is less than $d_{\rm lim}$, then the RSS can be considered in the near-field. In this paradigm, the impinging spherical wave forms a circular and divergent phase gradient on the RSS 
\textcolor{black}{area.} Accordingly, the RSS acts as an anomalous mirror and the two-hop link acts as a one-hop path. Hence, the distance path loss is affected by the summation of the traveled distances (i.e., Tx-RSS and RSS-Rx distances), which is known as the specular reflection paradigm \cite{Basar2019, Tang2019,Ntontin2019a}.

For a Tx-RSS (resp. RSS-Rx) distance $D \leq d_{\rm lim}$ and a large-sized RSS (LRSS), 
where the reflector unit sizes are $10 \lambda \times 10\lambda$ m$^2$, the minimum required number of reflectors for specular reflection, denoted $N_{\min}$, can be calculated as
\begin{equation}
\label{eq:Nmin}
D=\frac{2 A_t}{\lambda}= \frac{2 N_{\min} \left( 10 \lambda \right)^2}{\lambda}= 200 \lambda N_{\min} \Leftrightarrow N_{\min} = \dfrac{D}{200 \lambda}.
\end{equation}
The defined $N_{\min}$ will be later used to assess the feasibility of the RSS using the specular reflection paradigm.


In order to conduct the link budget analysis in this paradigm, we assume the Tx-Rx communication assisted by a building-mounted LRSS. Let $x(t)$ be the transmitted signal by the Tx. 
Then, the received (noise-free) signal at the Rx, denoted by $y(t)$, can be written as \cite{Basar2019} 
\begin{equation}
\label{eq:received}
y(t) = a \; x(t), 
\end{equation}
where $a$ is the wireless channel coefficient, 
expressed for the sake of simplicity with the log-distance path-loss model. The latter is given by
\begin{equation}
a = \sqrt{P_t G_t G_r} \left( \dfrac{\lambda}{4\pi d_0} \right) \left( \left(\dfrac{d_0}{d_l}\right)^\gamma + \sum_{i=1}^{N} \dfrac{ d_0^\gamma \rho_i e^{-j(\theta_i + \phi_i)}}{\left(d_{ti}+d_{ir}\right)^\gamma}\right), 
\end{equation}  
where $P_t$, $G_t$, and $G_r$ are the transmit power and the transmitter and receiver gains, respectively. Also, $d_0$ denotes the reference distance, $d_l$ is the distance between Tx and Rx,\footnote{We assume here that the direct link Tx-Rx is a weak link \cite{bjornson2021reconfigurable}.
} $2\gamma = \alpha$ is the path-loss exponent, and $N$ is the total number of RSS reflector units. Finally, $\rho_i$, $d_{ti}$, $d_{ir}$ \textcolor{black}{represent the reflection loss of the $i^{th}$ RSS reflector, the distance between Tx and RSS $i^{th}$ reflector, and distance between the RSS $i^{th}$ reflector and Rx, respectively. Also,  $\theta_i$ and $\phi_i$ are the corresponding incident and reflection angles.}  
\textcolor{black}{Now, the received power at Rx, denoted by} $P_r$, can be written as 
\begin{equation}
\label{eq3}
P_r =P_t G_t G_r \left( \dfrac{\lambda \; d_0^{\left(\gamma-1 \right)}}{4\pi }\right) ^2  \left( \dfrac{1}{d_l^\gamma} + \sum_{i=1}^{N}\dfrac{\rho_i e^{-j(\theta_i + \phi_i)}}{\left(d_{ti}+d_{ir}\right)^\gamma}\right)^2.
\end{equation}
For the sake of simplicity, we assume here that the LRSS 
can perfectly adjust the desired phase shifts, and that the reflectors are ideal without any \textcolor{black}{reflection loss, that is,}\footnote{\textcolor{black}{In practice, RSS reflection loss depends on  the configuration technology and building materials \cite{Hum2014,wu2020intelligent}. Also, since continuous phase shift implementation is difficult, only a finite discrete set of phase shifts is typically designed. It \textcolor{black}{has been shown}  that near-optimal RSS performance can be realized using a small number of phase-shift levels \cite{wu2019beamforming,huang2018energy}. In any case, 
\textcolor{black}{the loss of a few dB due to} the material properties or due to sub-optimal RSS configuration is insignificant compared to the signal propagation loss \cite{bjornson2021reconfigurable}. }}
\begin{equation}
\label{perfectshift}
\theta_i + \phi_i=0\; \text{and}\; \rho_i=1,  \; \forall i=1,\ldots,N.
\end{equation}
Moreover, assuming that the variation of $d_{ti}$ and $d_{ir}$ is negligible across the RSS, we have 
\begin{equation}
    \label{Eq:distApprox}
  d_{ti} +d_{ir} \approx 2d,\; \forall i=1,\ldots,N,
\end{equation}
where $d=d_l/2$.
Hence, the received power can be rewritten as \begin{equation}\label{Eq:Pr2}
P_r = P_t G_t G_r \left( \dfrac{\lambda }{4\pi }\right)^2 \left(\frac{d_0^{\left(\alpha-2\right)}}{(2d)^\alpha}\right) \left(1+N\right)^2. 
\end{equation}
According to (\ref{Eq:Pr2}), the location of the RSS has no impact on the received power, while the path loss is the dominant factor (e.g., $\alpha \geq 3$ in urban environments \cite{rappaport1996wireless}). Thus, in typical terrestrial environments, although the received power enhances quadratically with the number of RSS reflectors, it degrades at a much higher rate with 
the propagation distance.

\subsection{The Scattering Reflection Paradigm}
Assuming that the Tx-RSS and RSS-Rx distances are large (i.e., higher than $d_{\rm lim}$), and that tiny RSS reflector units are used (i.e., with dimensions between $0.1\lambda$ and $0.2\lambda$), then 
each reflector capturing the transmitted signal behaves as a new signal source that re-scatters the signal towards the UE. In this 
\textcolor{black}{paradigm, —the scattering reflection paradigm— the total effect} on the transmitted signal is the resultant of the cascaded individual channels Tx-RSS and RSS-Rx \cite{basar2020simris,ellingson2019path,Gab2020}. 

\textcolor{black}{A scattering paradigm for RSS-assisted communication is typically presented as an alternative to the degradation of direct links caused by strong blockages}
\cite{yildirim2020modeling}. Accordingly, the effect of the direct link is ignored, and the effective received signal is only the one scattered by the RSS. Since in this paradigm, the use of tiny RSS reflector units is advocated, we 
\textcolor{black}{name these small-sized RSS} (SRSS).

To accurately assess the SRSS-assisted terrestrial communications, we present next the link budget analysis for two channel models, namely the log-distance and 3GPP based models\footnote{
\textcolor{black}{Several researchers have recently raised practical concerns} about the specular reflection paradigm and the use of LRSS \cite{wu2020intelligent,myths}. Subsequently, we adopt \textcolor{black}{in this paper} the practical 3GPP channel model under the scattering reflection paradigm only.} 

\subsubsection{Log-Distance Channel Model}
For the $i^{th}$ reflector unit of the SRSS, the channel effect of the received signal, denoted $a_i$, is resulting from two cascaded channels, i.e., Tx-$i^{th}$ reflector and $i^{th}$ reflector-Rx  \cite{basar2020simris,ellingson2019path}. The channel coefficient is given by
\begin{equation}
\label{Eq:ref.ch.effect}
a_i = \sqrt{P_t G_{t} G_{r}} h_{ti}g_{ir} \rho_i e^{-j\phi_i},\; i=1,\ldots,N,
\end{equation}
where $\phi_i$ is the adjusted phase shift of the reflector, while $h_{ti}$ and $g_{ir}$ are the complex-valued coefficients representing the links between the $i^{th}$ reflector and both Tx and Rx. The latter are defined by
\begin{equation}
\begin{aligned}
h_{ti}=  \left( \dfrac{\lambda}{4\pi d_0}\right) \left( \dfrac{d_0}{d_{ti}}\right)^\gamma & e^{j \theta_{ti}}, \; \; \text{and}\;\;  g_{ir}=\left( \dfrac{\lambda}{4\pi d_0} \right) \left( \dfrac{d_0}{d_{ir}}\right)^\gamma  e^{j \theta_{ir}},
\end{aligned}
\end{equation} 
with $\theta_{ti}$ and $\theta_{ir}$ denoting the transmit and receive channel phases, respectively.
Following the generalization to the $N$ reflectors, the received power can be expressed by
\begin{equation}
P_r= P_t G_{t} G_{r} \left( \dfrac{\lambda}{4\pi d_0}\right)^4 d_0^{(2 \alpha)} \left( \sum_{i=1}^{N} \dfrac{\rho_i e^{-j  \left(\phi_i - \theta_{ti}-\theta_{ir}\right) }  }{(d_{ti}d_{ir})^\gamma} \right) ^2. 
\end{equation}
We assume lossless reflectors, i.e., 
\begin{equation}
\label{eq:rho}
    \rho_i =1, \forall i=1,\ldots,N,
\end{equation}
and that
\begin{equation}
\label{eq:distt}
    d_t \approx d_{ti} \; \text{and} \; d_r \approx d_{ir}, \forall  i=1,\ldots,N,
\end{equation}
where $d_t$ and $d_r$ are reference distances measured between the center of the SRSS and the Tx and Rx, respectively. Subsequently, the received power can be maximized by coherently combining the received signals through the $N$ reflectors, i.e., $\phi_i = \theta_{ti} + \theta_{ir}$. Hence, the received power can be written as
\begin{equation}
\label{Eq:terres_scattering}
P_r = P_t G_t G_r \left( \dfrac{\lambda}{4\pi}\right)^4   \left( \frac{d_0^{(2 \alpha-4)}}{\left(d_t d_r\right)^\alpha} \right)  N^2.
\end{equation}
According to (\ref{Eq:terres_scattering}), $P_r$ degrades faster 
than in (\ref{Eq:Pr2}), due to the distances multiplication. Also, $P_r$ is maximized when the RSS is the closest to either Tx or Rx.

\subsubsection{3GPP Channel Model}
RSS in terrestrial environments can be placed on facades of buildings to smartly reflect signals toward users. Since such smart buildings are expected to be available in modern urbanized environments, we assume in the following the urban scenario of the 3GPP standard model \cite{3gpp38901study}
\footnote{\textcolor{black}{Practically speaking, we envision that RSS to be deployed on the facades of high-rise buildings, which are available in urban environments. Also, we expect RSS to be implemented where they would bring profit
to service providers. Due to the high density of customers in urban areas, it is more likely that RSS will be deployed massively in urban environments and perhaps rarely or never in rural environments.}}.
The total path loss for the Tx-SRSS and SRSS-Rx links can be written as
\begin{equation}
\label{Eq:pathl}
   PL=\mathcal{P}^{\rm{LoS}} PL^{\rm{LoS}} + \mathcal{P}^{\rm{NLoS}} PL^{\rm{NLoS}} + PL_e, 
\end{equation}
where $\mathcal{P}^{\rm{LoS}}$ and $\mathcal{P}^{\rm{NLoS}}$ are the LoS and NLoS probabilities, $PL^{\rm{LoS}}$ and $PL^{\rm{NLoS}}$ are the associated losses in the LoS and NLoS conditions, and $PL_e$ accounts for the extra loss of indoor users. 
\textcolor{black}{The latter vary greatly in terms of building type,}
location within the building, and movement in the building. \textcolor{black}{For an accurate calculation of $PL_e$,}  we refer the reader to \cite{itu2020}. 
\textcolor{black}{But since the focus in our system is on outdoor users, $PL_e$ is ignored.}

The LoS probability in a terrestrial environment between the SRSS and Tx or Rx, assuming that the 
heights of the RSS-equipped building and Rx are below 13 m,
can be given by \cite[Table 7.4.2-1]{3gpp38901study}
\begin{equation}
\label{Eq;PLoS_ter}
    \mathcal{P}^{\rm{LoS}}=
    \begin{cases}
       1 \enspace &\rm{if} \enspace d_{2D} \leq 18 m\\
       \dfrac{18}{d_{2D}} + \exp \left(-\dfrac{d_{2D}}{63}\right) \left(1-\dfrac{18}{d_{2D}}\right)  \enspace &\rm{if} \enspace d_{2D} > 18 m,\\
        \end{cases}
\end{equation}
where $d_{2D}$ is the 2D separating distance (projected on the ground) between the SRSS and the Tx or Rx, whereas
the path loss for LoS and NLoS links is given as follows \cite[Table 7.4.1-1]{3gpp38901study}: \begin{equation}
\label{Eq:Los}
    PL^{\rm{LoS}} = 28+ 22 \log (d_{3D}) + 20 \log(f) + X
\end{equation}
and
\begin{equation}
\label{Eq:NLos}
PL^{\rm{NLoS}} = \max (PL^{\rm{LoS}}, \bar{PL}^{\rm{NLoS}})
\end{equation}
where
\begin{eqnarray}
\label{Eq:NLos2}
    \bar{PL}^{{\rm{NLoS}}}&=&    13.54 + 39.08 \log (d_{3D}) + 20 \log(f)\nonumber \\
     &-& 0.6 \; (H_x - 1.5)+ X,\; x \in \{ \text{RSS}, \text{Rx} \},
\end{eqnarray}
$d_{3D}$ is the 3D Tx-SRSS or SRSS-Rx separation distance in meters, $f$ is the carrier frequency in GHz, $X$ is a log-normal random variable denoting the shadow fading, with standard deviation $\sigma = 4$ dB and $\sigma = 7.8$ dB for LoS and NLoS links, respectively, and $H_{x}$ denotes the SRSS/Rx height. Specifically, for the Tx-SRSS link, $H_{x} = H_{RSS}$, where $H_{RSS}$ is the height of the building coated with the RSS, while for the SRSS-Rx link, $H_{x} = H_{Rx}$, with $H_{Rx}$ is the Rx height.

Accordingly, the received power, in dBm, can be written as
\begin{equation}
\label{eq:terres_3GPP}
    P_r = P_t +  G_t  +G_r - PL_{Tx-SRSS} - PL_{SRSS-Rx} + 20 \log{(N)},
\end{equation}
where $PL_{Tx-SRSS}$ and $PL_{SRSS-Rx}$ 
\textcolor{black}{represent the path loss} for the Tx-SRSS and SRSS-Rx links \textcolor{black}{respectively}, calculated using (\ref{Eq:pathl}).
\section{Link Budget Analysis for RSS-Assisted Non-Terrestrial Networks}\label{sec:budget_non_terr}

In this section, we start by presenting the aerial platforms characteristics. Then, we derive the link budget analysis of RSS-enabled non-terrestrial 
\textcolor{black}{communications (i.e., when an RSS is mounted on an aerial platform, such as a UAV, HAPS, or LEO satellite).}
Given the identified specular and scattering reflection paradigms,  we derive the related received power expressions.

\begin{figure}[t]
	\centering
	\includegraphics[width=\linewidth]{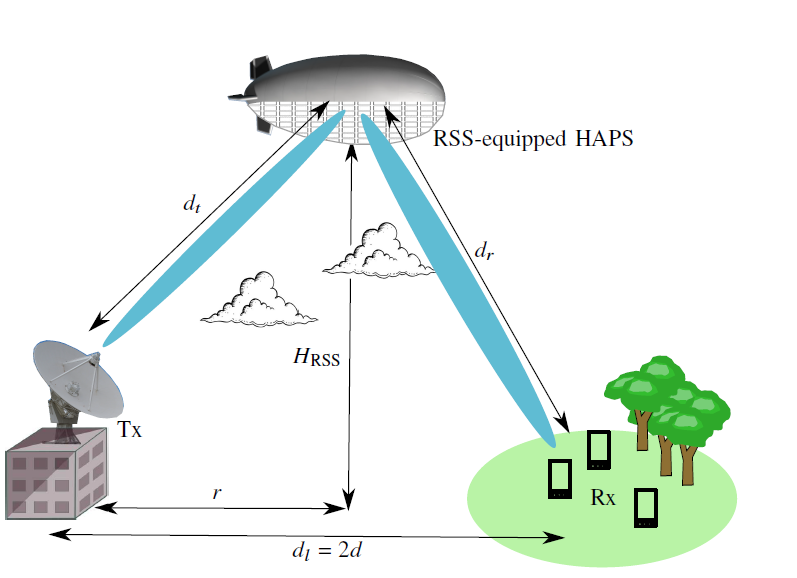}
	\caption{System model of an aerial platform equipped with RSS. 
	}
			\label{Fig:aerial_RSS}
\end{figure}
\subsection{Aerial Platforms Characteristics}\label{sec:plat_features}

\textcolor{black}{Three types of aerial platforms can be used to host RSS technology, namely UAVs, HAPS nodes, and LEO satellites.}
These platforms are usually deployed at different altitudes and target different coverage areas. Specifically, UAVs are deployed around 100 m altitude (and can operate up to 300 m \cite{3gpp2017enhanced}), to serve an area below 5 km \textcolor{black}{radius} \cite{alzenad20173,Mozaffari2019}.
\textcolor{black}{By contrast,} HAPS nodes are quasi-stationary platforms positioned in the stratosphere at an altitude between 17 and 50 km. However, most HAPS projects 
\textcolor{black}{target an altitude of 20 km} due to its preferred 
atmospheric characteristics for the platform stability and communication quality \cite{kurt2020vision}. Thus, HAPS systems have a much wider footprint than UAVs that spans from 40 to 100 km for high throughput \cite{kurt2020vision,sahabul2021} and can go up to 500 km according to the International Telecommunications Union (ITU) \cite{ITU_F1500}. Finally, LEO satellites 
orbit 
\textcolor{black}{the earth at an altitude between 400 and 2,000 km,} with an orbital period between 88 and 127 minutes \cite{3gpp2017Technical}. Accordingly, LEO satellites have the largest coverage footprint. Nevertheless, the coverage area significantly depends on the satellite's 
\textcolor{black}{altitude, elevation angle, and coverage scheme (i.e., whether it uses a spot or wide communication beam).}
Consequently, the LEO satellite's footprint has a radius between hundreds and thousands of kilometers \cite{fossa1998overview,cakaj2016coverage,su2019broadband}. To communicate with UAVs and HAPS nodes,  UE can use the same device 
\textcolor{black}{as it would to communicate with a terrestrial BS.}
However,  different equipment is needed to communicate with LEO satellites, since they require LEO tracking, either mechanically or electronically, in order to compensate for the satellite motion and achieve a reliable communication \cite{3gpp2017Technical}.

In addition to the aforementioned characteristics, the physical size of the aerial platform is crucial to enable hosting the RSS.  
UAVs have the smallest size, and 
 hence can dedicate only
a small area for RSS. On the other hand, two types of HAPS are identified, namely the aerostatic HAPS and the aerodynamic HAPS. \textcolor{black}{Aerostatic HAPS nodes, especially airships, are giant platforms whose lengths are typically between 100 and 200 m}, whereas aerodynamic HAPS nodes have wingspans between 35 and 80 m \cite{kurt2020vision}. Finally, the size of current LEO satellites is less than 10 m \cite{LEO_sat_iridium,LEO_sat_est}. Since HAPS nodes and LEO satellites are typically equipped with solar panel arrays, a part of the platform surface can be used to mount  RSS equipment.

\subsection{The Specular Reflection Paradigm}
When 
the Tx and Rx\footnote{Notice that Tx and Rx in the non-terrestrial communication context can be different according to the  aerial platform used. For instance, (Tx, Rx) can be (BS, UE) for a UAV, or it can be (Gateway, UE) \textcolor{black}{for a HAPS platform or a LEO satellite. A HAPS UE might be a mobile, vehicular, or fixed cellular user device, whereas a LEO has a fixed UE such as a household receiver \cite{starlink}.}} 
are separated by a relatively long distance, an aerial platform equipped with an LRSS can be used to assist the communication. Specifically, the Tx transmits its signal to the LRSS-equipped aerial platform. Then, the LRSS smartly reflects the incident signal towards the Rx, as illustrated in Fig. \ref{Fig:aerial_RSS}. 
Hence, the received noise-free signal in the specular paradigm is identical to the one in (\ref{eq:received}), whereas the channel effect is expressed using the free-space path-loss formula as
\begin{equation}
a = \sqrt{P_t G_t G_r} \left( \dfrac{\lambda}{4\pi}\right) \sum_{i=1}^{N}\dfrac{ \rho_i e^{-j(\theta_i + \phi_i)}}{d_{ti}+d_{ir}},
\end{equation}
where the LoS wireless link component is predominant, i.e., the path-loss exponent $\alpha=2 \gamma=2$.
By following the same assumptions as in (\ref{eq:rho})--(\ref{eq:distt}) and adopting a similar phase-shift configuration as in (\ref{perfectshift}), the received power can be given by
\begin{equation}
\label{Eq:aerial3}
P_r = P_t G_t G_r \left( \dfrac{\lambda}{4\pi}\right)^2  \dfrac{N^2}{(d_t + d_r)^2}. 
\end{equation} 
Unlike (\ref{Eq:Pr2}), 
the parameters $\lambda$, $N$, and $(d_t+d_r)$ have the same scaling law for the received power improvement or degradation.


One of the distinct features of aerial platforms compared to terrestrial networks is 
\textcolor{black}{their flexibility of placement, which allows to further enhance the system performance.}
\begin{Proposition}\label{prop:sp}
Given that Tx and Rx are collinearly separated by distance $d_l=2 d$, and that the aerial platform is located at altitude $H_{RSS}$ and horizontally separated from the Tx by a distance $r$, i.e., $d_t=\sqrt{H_{RSS}^2+r^2}$ and $d_r=\sqrt{H_{RSS}^2+(2d-r)^2}$, then the optimal placement of the aerial platform under specular reflection is given by 
\begin{equation}
r^*=d.
\end{equation}
\textcolor{black}{That is,} the platform is placed
over the perpendicular bisector of the segment Tx-Rx at altitude $H_{RSS}$. 
\end{Proposition}  
\begin{proof}
The specular equivalent path-loss distance, denoted $d_{sp}$, can be written as 
\begin{equation}\label{Eq:D_sp}
{d}_{sp} = d_t + d_r = 
\sqrt{(H_{RSS}^2 + r^2)}  + \sqrt{H_{RSS}^2+(2d-r)^2}. 	\end{equation} 
To maximize $P_r$, we need to minimize $d_{sp}$ through nulling its first derivative, i.e,
\begin{equation}\label{Eq:1st_dervative_sp}
\dfrac{\partial {d}_{sp}}{ \partial r} =\dfrac{r}{\sqrt{r^2+H_{RSS}^2}}-\dfrac{2d-r}{\sqrt{\left(2d-r\right)^2+H_{RSS}^2}}=0,
\end{equation}
from which we obtain $r^{*}=d$. When substituting this solution in the second derivative of $d_{sp}$, $\frac{\partial^2 d_{sp}}{\partial r^2}$, we get
\begin{eqnarray}
\label{Eq:2nd_dervative_sp}    \frac{\partial^2 d_{sp}}{\partial r^2}&=&\dfrac{H_{RSS}^2\left(\left(r^2+H_{RSS}^2\right)^\frac{3}{2}+\left(\left(2r-r\right)^2+H_{RSS}^2\right)^\frac{3}{2}\right)}{\left(\left(2d-r\right)^2+H_{RSS}^2\right)^\frac{3}{2}\left(r^2+H_{RSS}^2\right)^\frac{3}{2}}\nonumber \\
&\overset{(r=d)}{=}& \dfrac{2H_{RSS}^2}{(d^2+H_{RSS}^2)^ {\frac{3}{2}}}.
\end{eqnarray}
Since $\frac{\partial^2 d_{sp}}{\partial r^2} > 0$, then $r^*=d$ is the optimal value that minimizes ${d}_{sp}$.
\end{proof}
Accordingly, the maximal received power can be calculated using (\ref{Eq:aerial3}) for $r^*=d$ as follows:
 \begin{equation}
 \label{Eq:aerial_SP}
 P_r^* = P_t G_t G_r \left( \dfrac{\lambda}{4\pi}\right)^2 \dfrac{N^2}{4 ( {H_{RSS}^2+d^2})}.
 \end{equation}

\subsection{The Scattering Reflection Paradigm}
Similar to Section II, the scattering reflection paradigm is investigated for two channel models, namely the log-distance and the 3GPP models. 

\subsubsection{Log-Distance Channel Model}
Similarly to the SRSS-assisted terrestrial network, the received power at the Rx in the scattering paradigm can be expressed using (\ref{Eq:terres_scattering}) for $\alpha=2$, \textcolor{black}{thus:}
\begin{equation}
\label{Eq:budgetsc}
    P_r=P_t G_t G_r \left( \frac{\lambda}{4 \pi} \right)^4 \left(\frac{N}{d_t d_r}\right)^2.
\end{equation}

Unlike the specular paradigm, where the optimum aerial platform location is over the perpendicular bisector of segment Tx-Rx, the best platform location under the scattering paradigm is expected to be different due to the cascaded channel effect.

\begin{Proposition}\label{prop:sc}
Given that Tx and Rx are collineraly separated by distance $d_l=2d$, and that the aerial platform is located at altitude $H_{RSS}$ and horizontally separated from Tx by a distance $r$, then the optimal placement of the aerial platform under scattering reflection is given by
\begin{equation}
r^*=
\begin{cases}
d \pm \sqrt{d^2-H_{RSS}^2} \enspace &\text{\rm{if}} \enspace d\geq H_{RSS}\\
d \enspace  & \text{\rm otherwise}.
\end{cases}
\end{equation}

\end{Proposition}  
\begin{proof}
The scattering equivalent path-loss distance, denoted $d_{sc}$, can be expressed through
\begin{equation}
\label{Eq:D_sc}
    d_{sc}^2=(d_t d_r)^2=\left( H_{RSS}^2+r^2\right) \left( H_{RSS}^2 + \left( 2d-r\right)^2 \right).
\end{equation}
To maximize $P_r$, 
\textcolor{black}{we must minimize}
$d_{sc}^2$ by nulling its first derivative as follows:
\begin{equation}\label{Eq:1st_dervative}
\dfrac{\partial {d}^2_{sc}}{\partial r} =4r^3-12 d  r^2+\left(8 d^2+4H_{RSS}^2\right)r-4 d H_{RSS}^2=0.
\end{equation}
By solving (\ref{Eq:1st_dervative}), we obtain the follwing roots:
\begin{equation}\label{Eq:r_cases}
r^*= \begin{cases}
d + \sqrt{d^2-H_{RSS}^2} \\
d - \sqrt{d^2-H_{RSS}^2} \\
d.
\end{cases}
\end{equation}
The second derivative is given by
\begin{equation}\label{Eq:2nd_dervative}
\dfrac{\partial ^2 {d}^2_{sc}}{\partial r^2} =12r^2-24d r+8d ^2+4H_{RSS}^2,
\end{equation}
in which we substitute the roots of (\ref{Eq:r_cases}). Consequently, we find that (\ref{Eq:2nd_dervative}) is positive only when \{$r=d$ and $d \leq H_{RSS}$\} or \{$r=d \pm \sqrt{d^2-H_{RSS}^2}$ and $d \geq H_{RSS}$\}. These ($r,d$) values dictate the best aerial platform locations 
\textcolor{black}{where received power is maximal.}
\end{proof} 
Physically speaking, Proposition (\ref{prop:sc}) implies that under the scattering reflection paradigm, when the height of a platform is larger than its designed coverage radius, the platform should be placed at the \textcolor{black}{mid-point} between the Tx and Rx. However, when the targeted coverage radius is larger than the platform's height, two optimal locations can be used, which are being close to either the Tx or the Rx. 

By substituting the result of (\ref{Eq:r_cases}) in (\ref{Eq:D_sc}) and then (\ref{Eq:budgetsc}), the received power at Rx can be written as
\begin{equation}
\label{Eq:aerial_scattering}
P_r^* = P_t G_t G_r \left( \dfrac{\lambda}{4\pi}\right) ^4   \dfrac{N^2}{({d}^*_{sc})^2}. 
\end{equation}
where ${d}^*_{sc}$ is the equivalent path-loss distance for the optimal platform location, given by
\begin{equation}
{d}^*_{sc}=
\begin{cases}
H_{RSS}^2+d^2  \enspace &\text{if} \enspace d\leq H_{RSS}\\
2H_{RSS}d	     \enspace &\text{otherwise}.\\
\end{cases}
\end{equation}
  
Unlike (\ref{Eq:aerial3}), the impact of $H$, $d$, and $\lambda$, is more important than $N$ in (\ref{Eq:aerial_scattering}) due to the scattering effect.

\subsubsection{3GPP Channel Model}
For a fair comparison with RSS-assisted terrestrial communications, we investigate here the link budget analysis for RSS-enabled aerial platforms with realistic 3GPP channel models. In what follows, given that platforms operate at different altitudes and thus may experience different attenuation phenomena, 
\textcolor{black}{we study the link budget for RSS-assisted UAVs separately from RSS-equipped HAPS nodes and LEO satellites.}

\vspace{5pt}
\begin{itemize}
    \item \textbf{The UAV-Based Model:}
\end{itemize}

 RSS-enabled UAVs are envisioned to cooperate with terrestrial BSs to support terrestrial users. 
 \textcolor{black}{They can tackle coverage gaps }or increase the capacity of terrestrial users by reflecting the BSs' signals towards users \cite{alfattani2020aerial}. In such context, it is acceptable to consider the 3GPP channel model between terrestrial BSs and UAVs \cite{3gpp2017enhanced}. Due to the relatively low altitude of UAVs, the latter are generally assumed to have full LoS conditions. Consequently, the path loss in rural and urban environments\footnote{In contrast to RSS-assisted terrestrial communications, where RSS deployment is expected to be concentrated in urban areas, 
RSS-equipped aerial platforms 
\textcolor{black}{have the flexibility needed for both} urban and rural areas.} 
can be given by \cite[Table B-2]{3gpp2017enhanced} 
\begin{eqnarray}
\label{Eq:ruralm}
PL^{\rm rural}&=&
\max\left[23.9-1.8 \log_{10}(H_x) , 20\right] \log_{10}(d_{3D})\nonumber \\    &+&20\log_{10}(\dfrac{40 \pi f}{3}) +X, x \in \{ \text{UAV}, \text{UE} \}
\end{eqnarray}
and
\begin{equation}
\label{Eq:urbanm}
PL^{\rm urban}=28+  22\log_{10}(d_{3D})+20\log_{10}(f)+X, x \in \{ \text{UAV}, \text{Rx} \},
\end{equation}
where $H_x$, $d_{3D}$, $f$, and $X$ are defined as in (\ref{Eq:Los})--(\ref{Eq:NLos2}), except that $X$ has a log-normal distribution with standard deviation given by\textcolor{black}{\footnote{\textcolor{black}{Note that variable $X$ may be modified to include the effect of atmospheric turbulence, and thus the value of $\sigma$ can be adjusted accordingly.}}}
\begin{equation}
    \sigma=
    \begin{cases}
4.2 \;e^{(-0.0046 \; H_x)}  \enspace & \text{in rural environment,}  \\
 4.64 \;e^{(-0.0066 \; H_x)}	     \enspace & \text{in urban environment.}\\
    \end{cases}
\end{equation}
Subsequently, the received power in dBm can be expressed by
\begin{equation}
\label{Eq:UAV}
    P_r = P_t +  G_t  +G_r - PL_{Tx-UAV} - PL_{UAV-Rx} + 20 \log{(N)},
\end{equation}
where $PL_{Tx-UAV}$ and $PL_{UAV-Rx}$ are the path losses for the Tx-UAV and UAV-Rx links, calculated using either (\ref{Eq:ruralm}) or (\ref{Eq:urbanm}), depending on the considered environment.

\vspace{+5pt}
\begin{itemize}
    \item \textbf{The HAPS/LEO based Model:}
\end{itemize}


The support of future wireless networks through HAPS and LEO satellites is envisioned for both rural and urban areas \cite{kurt2020vision,kodheli2020satellite}.
To assess the performance in these different environments,  LoS probabilities are required. Based on the elevation angle of the aerial platform \textcolor{black}{relative to} the terrestrial Tx or Rx, denoted by $\vartheta$,  LoS probabilities can be estimated for different environments using \cite[Table 6.6.1-1]{3gpp2017Technical}. For the sake of simplicity, we propose to substitute \cite[Table 6.6.1-1]{3gpp2017Technical} by a LoS probability function, defined as
\begin{equation}
\label{Eq:Plos_model}
    \mathcal{P}^{LoS} = b_1 \vartheta^{b_2}+b_3,
\end{equation}
where $b_i$, $i=\{1,2,3\}$ are the parameters that depend on the environment,  determined in Table \ref{tab:LSO_prob}. 
The accuracy of (\ref{Eq:Plos_model}) is validated in Fig. \ref{fig:HAPS_LOS_prob}, where it is shown to agree with the results of \cite[Table 6.6.1-1]{3gpp2017Technical}.
\textcolor{black}{Now, the path loss for LoS and NLoS conditions, denoted by $PL$, can be written as}
\begin{equation}
\label{Eq:PLhaps}
    PL=\mathcal{P}^{\rm LoS} PL^{\rm LOS} + \mathcal{P}^{\rm NLoS} PL^{\rm NLOS}, 
\end{equation}
where $\mathcal{P}^{y}$ and $ PL^{y}$, $y \in \{ \text{LoS}, \text{NLoS} \}$, are defined as in (\ref{Eq:pathl}). 
The signal path between a HAPS/LEO and a terrestrial Tx or Rx undergoes several stages of propagation and attenuation. Specifically, the path losses $PL^y$, $y \in \{ \text{LoS}, \text{NLoS} \}$ are composed as follows \cite{3gpp2017Technical}:
\begin{equation}
PL^y= PL_{b}^y + PL_{g} + PL_{s} + PL_{e}, y \in \{ \text{LoS}, \text{NLoS} \},
\end{equation}
where $PL_{b}^y$ is the basic path loss, $PL_{g}$ is the attenuation due to atmospheric gasses, $PL_{s}$ is the attenuation due to either ionospheric or tropospheric scintillation, and
$PL_e$ is the building entry loss, expressed in dB. 

\begin{table}[t]
	\centering
	\caption{Parameters of HAPS/LEO LoS probability model (\ref{Eq:Plos_model}) }\label{tab:LSO_prob}
\begin{tabular}{|c|c|c|c|} 
	\hline 
	Parameter & Dense Urban & Urban & Rural \\ 
	\hline 
	$b_1$ & 0.04235 & 9.668  & -99.95 \\ 
	\hline 
	$b_2$ & 1.644 & 0.547 & -0.5895 \\ 
	\hline 
	$b_3$ & 27.32 & -10.58 & 104.1 \\ 
	\hline 
\end{tabular} 
\end{table}

\begin{figure}
    \centering
    \includegraphics[width=0.9\linewidth]{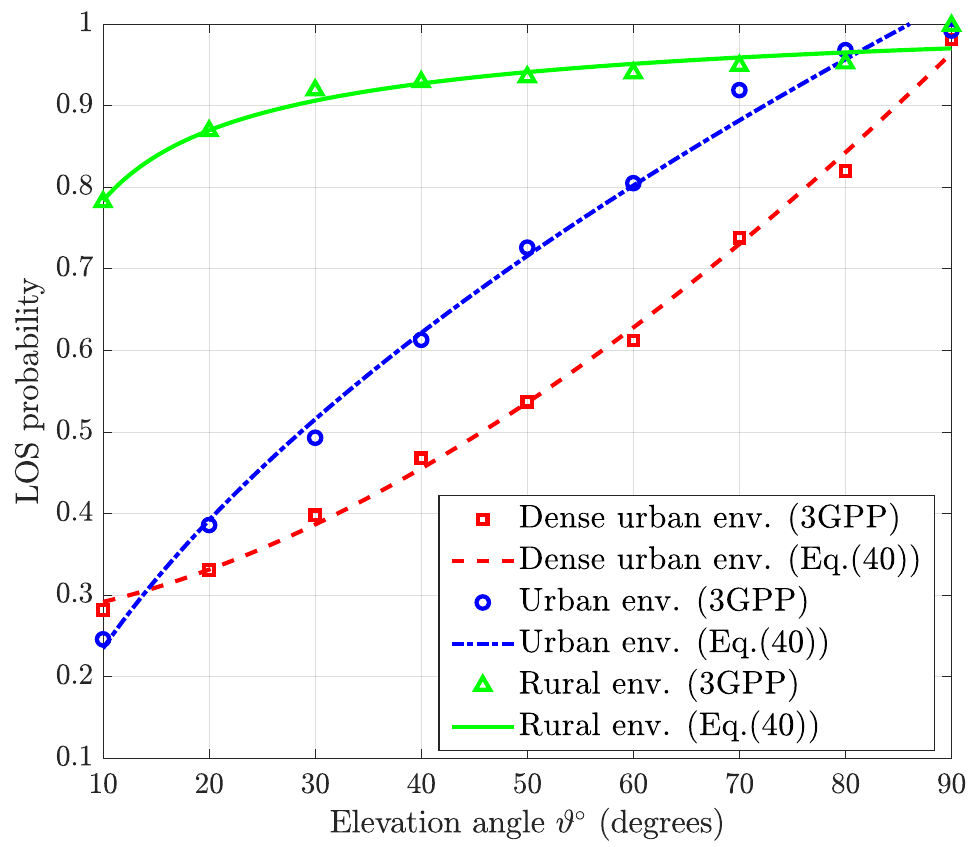}
    \caption{LoS probability of HAPS/LEO \textcolor{black}{relative to} elevation angle (different environments).}
    \label{fig:HAPS_LOS_prob}
\end{figure}

$PL_{b}^y$ accounts for the signal's free-space propagation ($FSPL$), clutter loss ($CL^y$), and shadow fading ($X^y$), i.e., 
\begin{equation}
    PL_b^y=FSPL+CL^y+X^y, \; y \in \{ \text{LoS}, \text{NLoS} \},
\end{equation}
where 
\begin{equation}
FSPL= 32.45+ 20\log_{10}(f) + 20\log_{10}(d_{3D}),
\end{equation}
with $d_{3D}$ being the 3D distance between the HAPS/LEO and the terrestrial Tx or Rx, expressed as a function of the HAPS/LEO altitude $H_z$, $z \in \{\text{HAPS},\text{LEO} \}$ and the platform's elevation angle $\vartheta$ as follows: 
\begin{equation}
    d_{3D}=\sqrt{R^2_E \sin^2({\vartheta})+ H_z^2 +2H_z R_E} - R_E \rm{sin} (\vartheta),
\end{equation}
where $R_E$ denotes \textcolor{black}{the earth's radius.}
The clutter loss, $CL^y$, represents the attenuation caused by buildings and environmental objects. Its value depends on $\vartheta $, $f$, and the environment type. 
In LoS conditions, $CL^{\rm LoS}=0$, while for NLoS, the $CL^{\rm NLoS}$ values of \cite[Tables 6.6.2-1 to 6.6.2-3]{3gpp2017Technical} can be used for the typical Ka spectrum band (i.e., between 26.5 and 40 GHz). 
Finally, $X^y$  
is a zero-mean normal distribution with standard deviation $\sigma^y$, $y \in \{ \text{LoS},\text{NLoS}\}$, whose values are determined in \cite[Tables 6.6.2-1 to 6.6.2-3]{3gpp2017Technical}.    
For the sake of simplicity, we present the average 
\textcolor{black}{values of the parameters} in Table \ref{tab:CL_sigma}\textcolor{black}{\footnote{\textcolor{black}{To be noted that variable $X^y$ can be modified to include the effect of external factors, such as the atmospheric turbulence, and platform's drift. Hence, the value of $\sigma^y$ can be modified accordingly.}}}.

\begin{table}[t]
	\centering
	\caption{Average clutter loss and shadow fading standard deviation in the Ka-band }\label{tab:CL_sigma}
\begin{tabular}{|c|c|c|c|} 
	\hline 
	Parameter & Dense Urban & Urban & Rural \\ 
	\hline 
	$CL^{NLoS}$ & 38.6 & 38.6  & 23.15 \\ 
	\hline 
	$\sigma^{LoS}$ & 1.75 & 4 & 1.15 \\ 
	\hline 
	$\sigma^{NLoS}$ & 14.7 & 6 & 10.75 \\ 
	\hline 
\end{tabular} 
\end{table}

$PL_g$ is the attenuation caused by absorption due to atmospheric gases. Its value depends mainly on $f$, $\vartheta$, and $d_{3D}$. According to \cite{3gpp2017Technical}, 
\textcolor{black}{the effect of atmospheric gases is negligible  for $f \leq 10$ GHz.}
However, for higher frequency bands 
suitable for  RSS operations with a large number of reflectors, the selection of frequency windows with minimum atmospheric effect is important.
In addition to the aforementioned factors, $PL_g$ depends on the dry air pressure $p$, water-vapor density $\xi$, and temperature $T$ \cite{sector2013recommendation}. An illustrative example is shown in Fig. \ref{fig:gas_loss} for different Tx-Rx link lengths, selected for typical distances between a HAPS and Tx or Rx (20 and 100 km) and between a LEO satellite and Tx or Rx (1,000 km). The related parameters ($p$, $\xi$ and $T$) are selected on the basis of the mean annual global reference atmosphere, i.e., $p= 101300$ Pa, $\xi = 7.5 g/m^3$, and $T=15^\degree$C \cite{itu1999p}. The $PL_g$ calculation is carried out using the steps detailed in \cite{sector2013recommendation}. 
\begin{figure}
    \centering
    \includegraphics[width=0.94\linewidth]{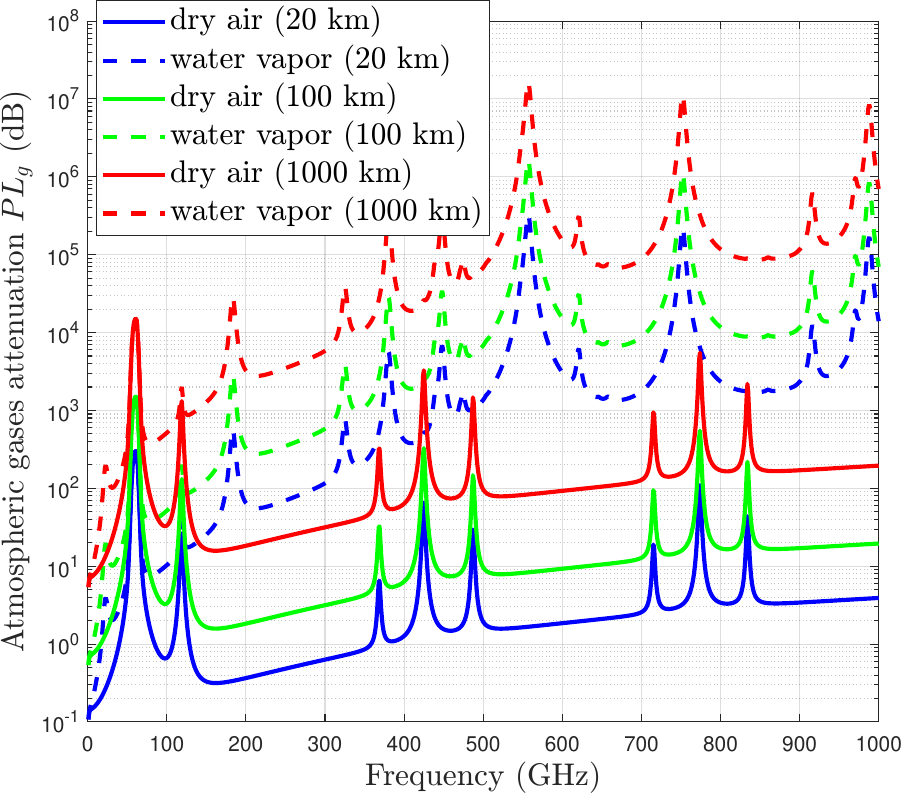}
    \caption{Attenuation due to atmospheric gases $PL_g$ vs. frequency (different path lengths). }
    \label{fig:gas_loss}
\end{figure}


Finally, the scintillation loss $PL_s$ is caused by rapid fluctuations of the received signal amplitude and phase. There are two types of scintillation losses, namely the ionospheric scintillation and the tropospheric scintillation. The 
\textcolor{black}{former only significantly disrupts signals at frequencies below 6 GHz,}
whereas the latter affects only signals in frequencies above 6 GHz.
The impact of the ionospheric scintillation is only significant for latitudes  in the range $[-20\degree, 20\degree]$ 
\cite{3gpp2017Technical}. \textcolor{black}{The ionospheric scintillation loss} can be calculated as detailed in \cite{Ionospheric2019ITU}. Specifically, \textcolor{black}{it} is derived from the measured peak-to-peak fluctuation as follows: 
\begin{equation}
    PL_s=\dfrac{{PF}}{\sqrt{2}},
\end{equation}
where $PF$ is the peak-to-peak fluctuation, equal to $1.1$ dB for $f=4$ GHz, and calculated by the following equation for $f \leq 6$ GHz:
\begin{equation}
    {PF}_{(f \leq 6\; \text{GHz})}={PF}_{(f=4\; \text{GHz})} \cdot (f/4)^{-1.5}.
\end{equation}

Since high frequency bands are advocated for the RSS-enabled HAPS/LEO platforms, tropospheric scintillation is taken into account in the path-loss model.
Specifically, the wireless signal fluctuations are due to sudden changes in the refractive index caused by temperature, water vapor content, and barometric pressure variations. Also, low elevation angles (especially below $5^\degree$) are significantly affected by the scintillation loss, due to the longer path of the signal and the wider beam width. 
The value of the tropospheric scintillation loss is 
\textcolor{black}{season and region dependent.} We refer the reader to \cite{series2015propagation} for the detailed calculation steps.
\textcolor{black}{To give an idea of the typical power attenuation level, }
the  tropospheric attenuation with 99\% probability at 20 GHz in Toulouse, France is tabulated in \cite[Table 6.6.6.2.1-1]{3gpp2017Technical} for different elevation angles. For the sake of simplicity, we model the related tropospheric scintillation loss by
\begin{equation}
\label{Eq:Ps_model}
    PL_s = 14.7 \; \vartheta^{(-1.136)},\; \vartheta \in [0^{\circ},90^{\circ}],
\end{equation}
which is a valid approximation of \cite[Table 6.6.6.2.1-1]{3gpp2017Technical} as shown in Fig. \ref{fig:tropos_sci_loss}.
\begin{figure}
    \centering
    \includegraphics[width=0.9\linewidth]{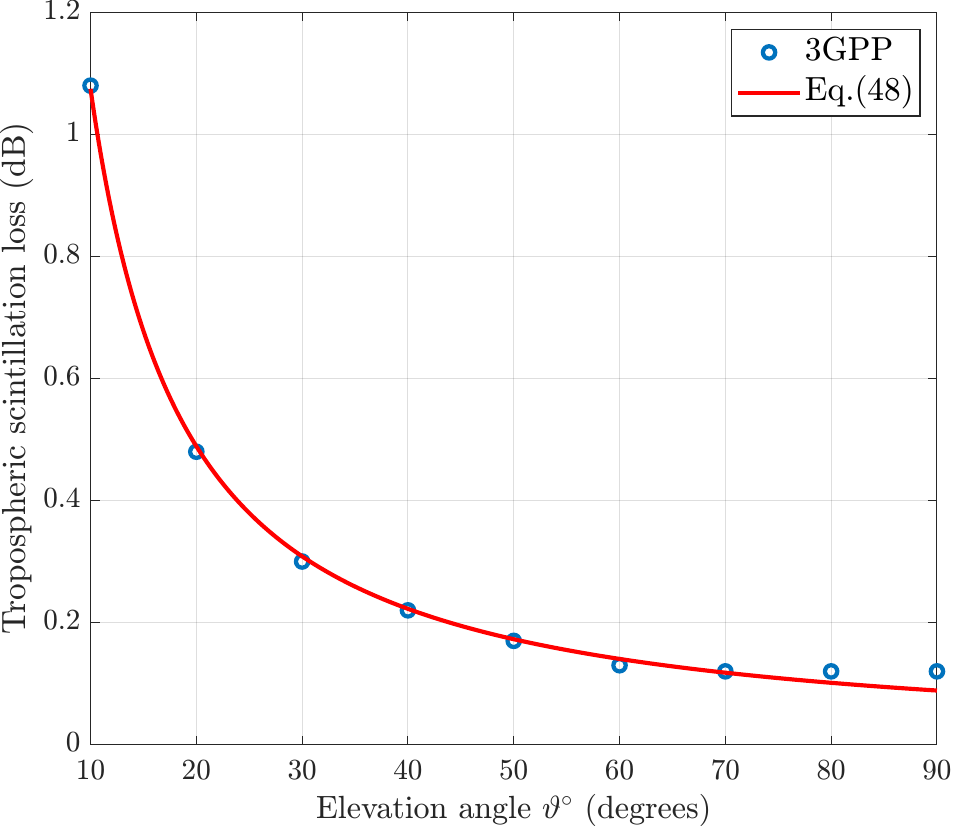}
    \caption{Scintillation loss $PL_s$ vs. elevation angle.}
    \label{fig:tropos_sci_loss}
\end{figure}

Subsequently, the received power in dBm at Rx can be expressed by
\begin{eqnarray}
\label{eq:HAPS_LEO_3GPP}
    P_r &=& P_t +  G_t  +G_r - PL_{Tx-H_L} -PL_{H_L-Rx} \nonumber \\ 
    &+&  20 \log{(N)}, \; H_L \in \{HAPS, LEO \},
\end{eqnarray}
where $PL_{Tx-H_L}$ and $PL_{H_L-Rx}$ are calculated using (\ref{Eq:PLhaps}).


The link budget analysis of Section III is summarized in Table \ref{Tab1}.
\begin{table*}
\centering
\caption{Link Budget Summary}
\label{Tab1}
\footnotesize
\begin{tabular}{|c|c|c|c|c|}
\hline
{\textbf{Scenario}} & {\textbf{Size of RSS reflector}} & \textbf{Channel Model} & {\textbf{Link Budget}} & {\textbf{Dominant parameter(s)}}  \\
\hline 
\hline
\makecell{LRSS-assisted terrestrial communication \\ {(Specular paradigm)}} & $10 \lambda \times 10 \lambda$ & \makecell{Log-distance} & Eq. (\ref{Eq:Pr2}) & \makecell{ $d$, $\lambda$ and $N$ (for $\alpha=2$) \\ $d$ (for $\alpha > 2$)}
 \\ \hline
\makecell{SRSS-assisted terrestrial communication \\ {(Scattering paradigm)}} & \makecell{$[0.1\lambda \times 0.1 \lambda, 0.2 \lambda \times 0.2 \lambda]$} & \makecell{Log-distance \\ {} \\ 3GPP}  & \makecell{Eq. (\ref{Eq:terres_scattering}) \\ {} \\ Eq. (\ref{eq:terres_3GPP})} & \makecell{$\lambda$ (for $\alpha < 4$) \\ $\lambda$, $d_t$, and $d_r$ (for $\alpha = 4$)\\ $f$}
 \\ \hline
\makecell{LRSS-assisted non-terrestrial communication \\ {(Specular paradigm)}} & $10 \lambda \times 10 \lambda$ & Log-distance & Eq. (\ref{Eq:aerial_SP}) & \makecell{$\lambda$, $N$, $H_{RSS}$ and $d$}
\\ 
\hline
\makecell{SRSS-assisted non-terrestrial communication \\ {(Scattering paradigm)}} & \makecell{$[0.1\lambda \times 0.1 \lambda, 0.2 \lambda \times 0.2 \lambda]$} & \makecell{Log-distance \\ 3GPP (UAV)\\ 3GPP (HAPS/LEO)}  & \makecell{Eq. (\ref{Eq:aerial_scattering}) \\  Eq. (\ref{Eq:UAV}) \\ Eq. (\ref{eq:HAPS_LEO_3GPP})} & \makecell{$\lambda$ \\ $f$ \\ $f$}
\\ 
\hline
\end{tabular}
\end{table*}

\section{Results and Discussion}\label{sec:results}
In this section, we evaluate the received power for different RSS-mounted platforms, and the impact of several parameters is investigated. Based on the unique features of each platform, as discussed in Sections II and III, we assume typical values of altitude, coverage radius, and RSS area for each platform, as shown in Table \ref{tab:platforms_chr}. Moreover, we consider the height of the terrestrial Tx $H_{\rm Tx}=25$ m, while the receiver's height is $H_{\rm Rx}=1.5$ m.

\begin{table}[t]
\centering
\caption{Typical parameters of different types of platforms.}\label{tab:platforms_chr}
{\renewcommand{\arraystretch}{1.5}
\begin{tabular}{|c|c|c|c|}
\hline
\textbf{\makecell{Platform}} & \textbf{\makecell{Altitude\\ ($H_{RSS}$)}} & \textbf{\makecell{Coverage radius\\ ($d$)}} & \textbf{ \makecell{RSS area\\ ($A_t$)}} \\
\hline
\hline
Terrestrial & 5 m &  0.5 km & {5 $\times$ 10 m$^{2}$} \\ 
\hline
UAV & 200 m & 2 km & 0.25 $\times$ 0.25 m$^{2}$ \\
\hline
HAPS & 20 km & 50 km & {40 $\times$ 20 m$^{2}$}  
\\
\hline
LEO & 500 km & 500 km & 5 $\times$ 10 m$^{2}$ \\
\hline
\end{tabular}
}
\end{table}

\subsection{Impact of the Number of Reflectors and Environment Type}
Given the selected platform, the type of mounted RSS reflectors is of great concern, since different communication paradigms can be 
\textcolor{black}{followed and these have different costs.}
According to \cite{ellingson2019path}, the cost is expected to be approximately proportional to the RSS size and number of RSS reflector units. Consequently, larger reflectors, with dimensions above $10 \lambda \times 10 \lambda$ and operating in the specular reflection paradigm, are more expensive than small reflectors (i.e., with dimensions lower than $0.2 \lambda \times 0.2 \lambda$), 
which operate in the scattering reflection paradigm. 
Hence, the maximal number of reflectors to install on a platform, denoted $N_{\max}$, is limited by the reserved area on the platform for the RSS $A_t$ and the reflector's size $A_r=c_1 \lambda \times c_2 \lambda$, where $c_1 \lambda>0$ and $c_2 \lambda>0$ are the length and width of a reflector unit, respectively. Their relation is defined by
\begin{equation}
\label{eq:Nmax}
    N_{\max}= \dfrac{A_t}{A_r}=\dfrac{A_t}{c_1 c_2 \lambda^2}.
\end{equation}

In order to  highlight the potential gains of using RSS-equipped platforms, we consider that a communication between a terrestrial Tx and a terrestrial Rx is assisted by an RSS-equipped platform, namely a building facade, a UAV, a HAPS, or a LEO satellite, characterized as in Table \ref{tab:platforms_chr}. For the sake of simplicity, we assume that Tx and Rx are located at the edges of the platform's coverage footprint (i.e., within a distance $2d$). Moreover, due to the flexibility of aerial platforms, we assume that they are placed at the optimal location according to the considered reflection paradigm, while the terrestrial RSS is assumed 
\textcolor{black}{to be midway} between Tx and Rx\footnote{This assumption for the terrestrial RSS is justified by the fact that the latter cannot be 
\textcolor{black}{moved later to another location} given that Tx and Rx locations may change.}.
The operating frequency $f={\lambda}/{c}=30$ GHz (Ka-band), where $c$ is the light's velocity in m/s, the receive antenna gain $G_r=1$ (i.e., 0 dBi), and the path-loss $\alpha=4$ for the terrestrial communication.
Finally, the following transmit/receive parameters are set as follows for the log-distance channel model: transmit power $P_t=40$ dBm and transmit antenna gain $G_t=1$ (i.e., 0 dBi). For the 3GPP model, $P_t$ and $G_t$ are fixed according to the related standards. Specifically, for the terrestrial and UAV systems, we assume that $P_t=35$ dBm and $G_t=8$ dBi \cite{3gpp38901study}, while for the HAPS and LEO systems, we set $P_t=33$ dBm and $G_t=43.2$ dBi \cite{3gpp2017Technical}. 

\begin{figure}[t]
	\centering
	\includegraphics[width=0.9\linewidth]{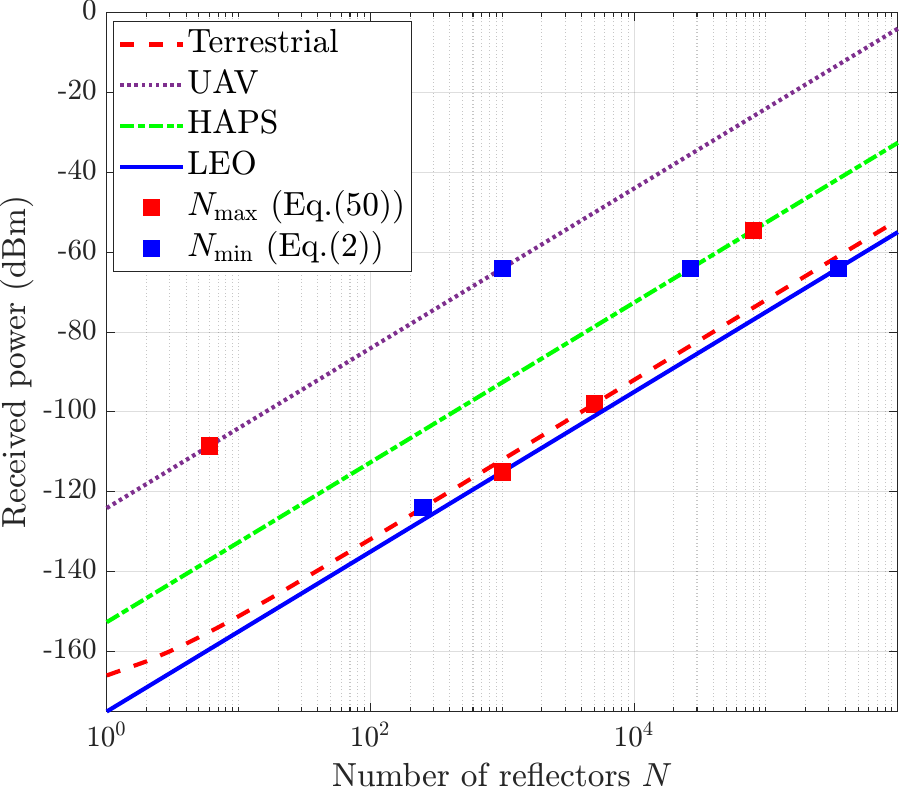}
	\caption{Received power vs. number of reflectors (different platforms; specular reflection paradigm).}
	\label{Fig:specular}
\end{figure}

Given the specular reflection paradigm, Fig. \ref{Fig:specular} presents the received power as a function of the number of reflectors for different platforms.
We notice that for UAVs and LEO satellites, the minimum required number of reflectors for specular reflection is greater than the maximum number of reflectors that can be placed on the platform's surface. Therefore, specular reflection cannot be achieved for RSS-equipped UAVs and LEO satellites for the specified coverage areas. This is due to the 
\textcolor{black}{limited area available for RSS on UAVs}
and to the relatively high operating altitude of LEO satellites.  
However, specular reflection can be realized using RSS-equipped HAPS systems or terrestrial environments. 
\textcolor{black}{This is because of the clear LoS links with HAPS nodes and the relatively short}
communication distances in terrestrial environments. 
\textcolor{black}{As we can also see in Fig. \ref{Fig:specular},} $N_{\min}=27,000$ reflectors for HAPS provides coverage in a 100 km circular area at $P_r=-64$ dBm. For a smaller coverage footprint, it is expected that a lower number of reflectors would be used.  Although a small number of reflectors is required in  terrestrial environment for  specular reflection, the received power is worse than in the HAPS scenario.
This is due to degraded terrestrial communication channels, compared to LoS wireless links for the RSS-equipped HAPS.
Accordingly, a HAPS is the preferred RSS mounting platform in the specular reflection paradigm.

\begin{figure}
    \centering
    \includegraphics[width=0.9\linewidth]{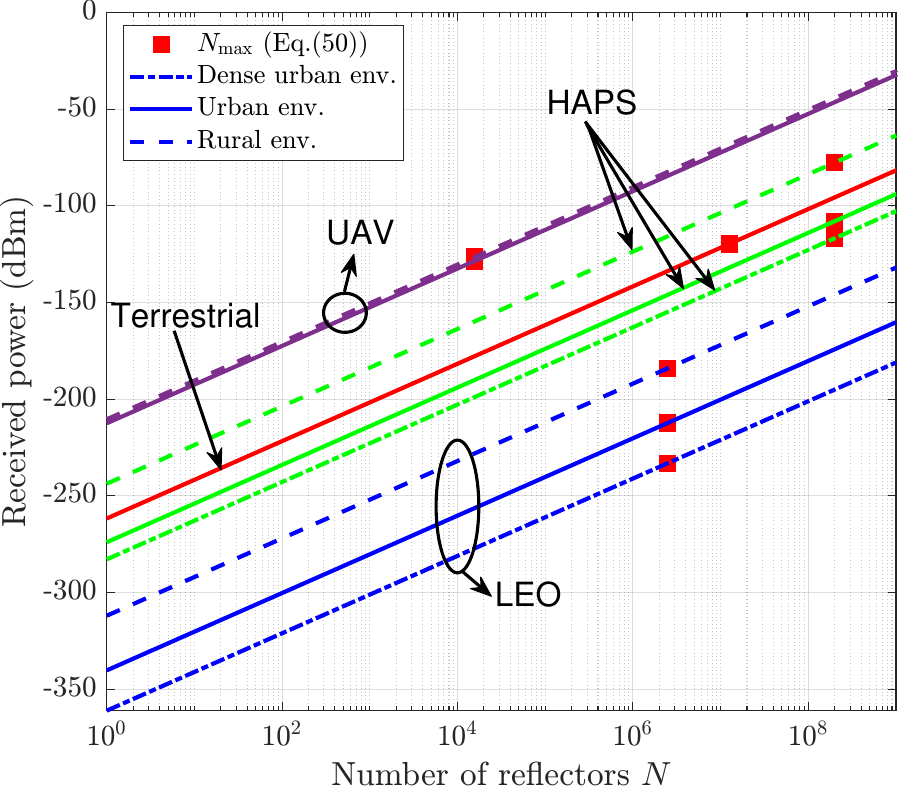}
    \caption{Received power vs. number of reflectors (different platforms and environment types; scattering reflection paradigm).}
    \label{fig:3GPP_PR_50_dry}
\end{figure}

For the scattering reflection paradigm, we present in Fig. \ref{fig:3GPP_PR_50_dry} the received power as a function of the number of reflectors for different platforms. 
\textcolor{black}{For the sake of simplicity,} we show only the results corresponding to the 3GPP channel models where different types of environments are studied, namely dense urban, urban, and rural. Given the same number of reflectors, the RSS-equipped UAV system achieves the best power performance due to short distances and clear LoS wireless links. Moreover, 
\textcolor{black}{the difference in performance between}
the rural and urban environments 
\textcolor{black}{is minimal, at about 2 dB.}
For a number of reflectors close or equal to $N_{\max}$,
the LEO system has the worst received power performance, 
\textcolor{black}{which degrades significantly in accordance with the density of the urban environment.}
In contrast, the HAPS system realizes the best $P_r$ values for a number of reflectors near or equal to $N_{\max}$. 
\textcolor{black}{This remains the case even in dense urban environments, where a HAPS system may compensate for performance degradation by using a higher number of RSS reflectors or by reducing its coverage footprint (i.e., where the Tx and Rx are closer).}
Consequently, 
\textcolor{black}{it is worth noting that RSS-equipped LEO systems}
may require further redesigning in order to be feasible, while the HAPS coverage footprint needs to be adjusted according to the environment type of the served area.

\subsection{Impact of the Carrier Frequency}
As shown in Table \ref{Tab1}, the carrier frequency is a dominant parameter in the link budget analysis. Although using higher frequencies enables high capacity links and addresses the spectrum scarcity issues, 
\textcolor{black}{such frequencies as millimeter wave and Terahertz may suffer from significant signal attenuation.}
Nevertheless, 
\textcolor{black}{high frequencies enable the use of small-sized reflectors. And since a large number of these can be used in small areas, this may counterbalance the signal attenuation.}

\begin{figure}[t]
	\centering
	\includegraphics[width=0.9\linewidth]{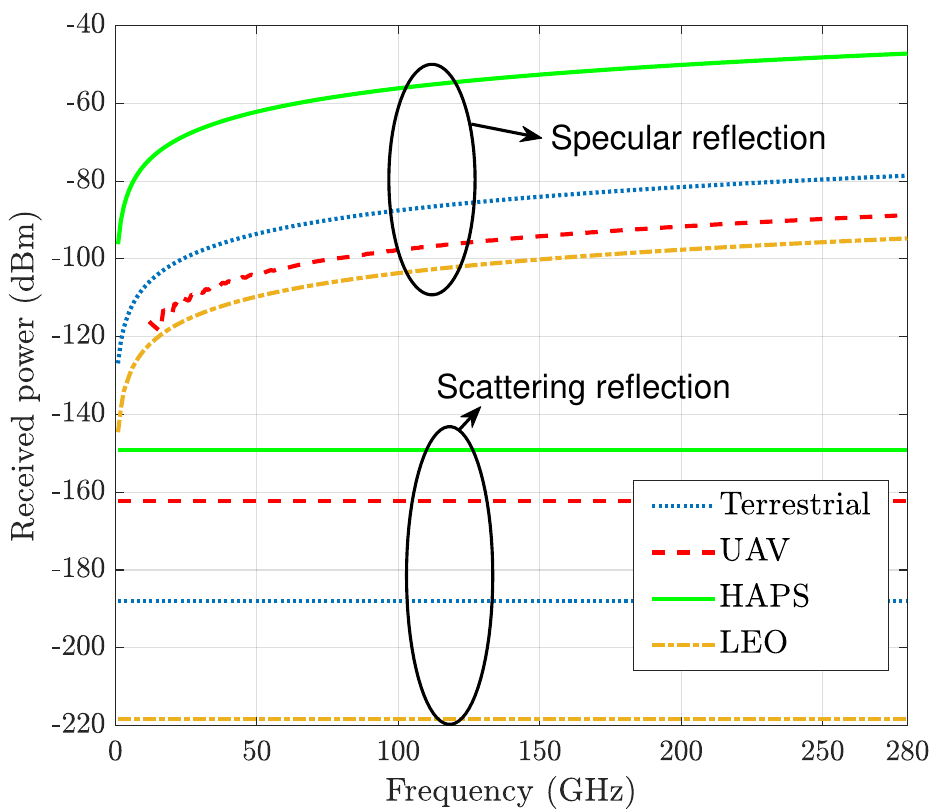}
	\caption{Received power vs. carrier frequency (different platforms).}
	\label{Fig:freq_realtion_combined}
\end{figure}

In Fig. \ref{Fig:freq_realtion_combined}, we evaluate the received power as a function of the carrier frequency, for the same communication model introduced previously. Here, we assume that each platform hosts the maximal number of reflectors $N_{\max}$, and that the link budget analysis for the scattering reflection paradigm is realized for the log-distance channel model 
\textcolor{black}{for the sake of comparison.}
The results show that in both specular and scattering reflection paradigms, the HAPS system provides the best performance due to its large surface area that accommodates the highest number of RSS reflectors. On the other hand, the RSS-equipped UAV performs worse than the terrestrial RSS under the specular reflection paradigm because of its small surface area, which accommodates a very small number of reflectors. In particular, when $f \leq 12$ GHz, a UAV with dimensions of $0.25 \times 0.25$ m$^2$, cannot host a single reflector. Given the scattering reflection paradigm, all platforms demonstrate stable performance for any carrier frequency. Indeed, by combining (\ref{eq:Nmax}) into (\ref{Eq:terres_scattering}) and (\ref{Eq:aerial_scattering}) respectively, we obtain the maximal received power, expressed by
\begin{eqnarray}
    \label{eq:terres_max}
    P_r^{\max}&=&P_t G_t G_r \left( \dfrac{\lambda}{4 \pi} \right)^4 \dfrac{A_t^2}{\left(c_1 c_2 \lambda^2\right)^2 \left( d_t d_r \right)^\alpha}\nonumber \\
    &=& \dfrac{P_t G_t G_r}{\left(4 \pi\right)^4} \left(\dfrac{A_t}{c_1 c_2 }\right)^2 \dfrac{1}{\left(d_t d_r \right)^\alpha},
\end{eqnarray}
for the terrestrial environment.  However, for the non-terrestrial one, (if $d\leq H_{RSS}$) $P_r^{\max}$ is expressed as
\begin{eqnarray}
    \label{eq:aer_max}
    P_r^{\max}&=&P_t G_t G_r \left( \dfrac{\lambda}{4 \pi} \right)^4 \dfrac{A_t^2}{\left(c_1 c_2 \lambda^2\right)^2 \left( H_{RSS}^2+d^2 \right)^2} \\
    &=& \dfrac{P_t G_t G_r}{\left(4 \pi\right)^4} \left(\dfrac{A_t}{c_1 c_2 }\right)^2 \dfrac{1}{\left(H_{RSS}^2 +d^2 \right)^2},\; 
    \nonumber
\end{eqnarray}
otherwise (i.e., $d>H_{RSS}$),
\begin{equation} \label{eq:aer_max2}
    P_r^{\max} = \dfrac{P_t G_t G_r}{\left(4 \pi\right)^4} \left(\dfrac{A_t}{c_1 c_2 }\right)^2 \dfrac{1}{\left(2H_{RSS}d\right)^2}.
\end{equation}
According to (\ref{eq:terres_max})--(\ref{eq:aer_max2}), the received power is no longer dependent on the frequency (or the wavelength), but rather on $A_t$, $c_1$ and $c_2$.


\begin{figure}
    \centering
    \includegraphics[width=0.9\linewidth]{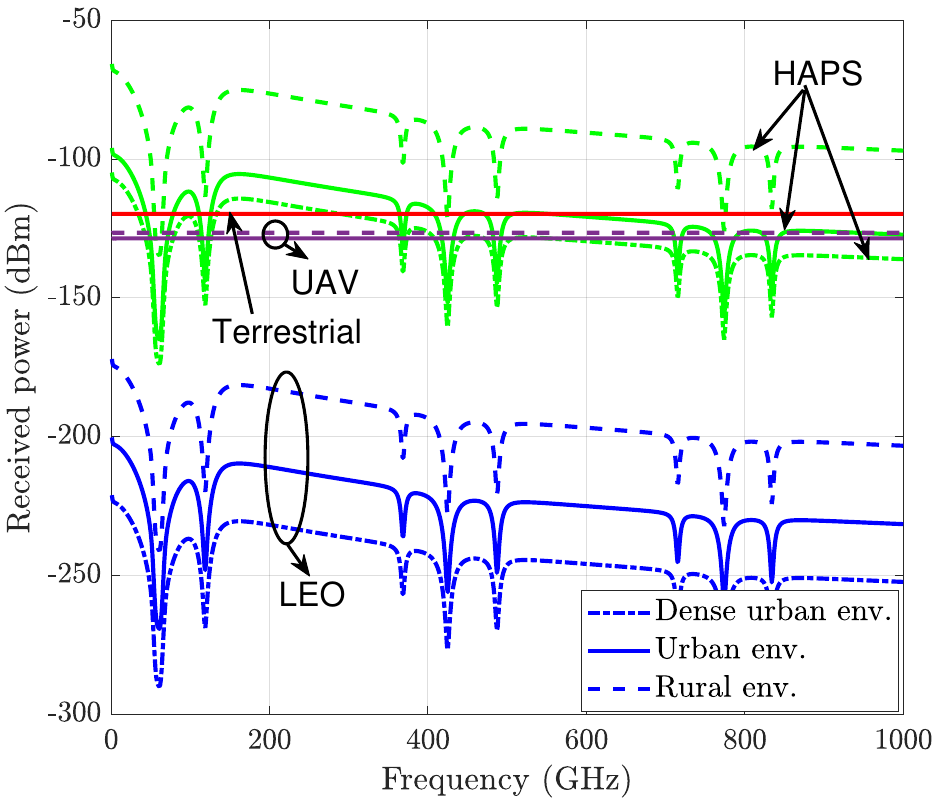}
    \caption{Received power vs. carrier frequency (different platforms and environment types).}
    \label{fig:3GPP_freq_20_dry}
\end{figure}

Fig. \ref{fig:3GPP_freq_20_dry} shows the received power as a function of the carrier frequency, given different RSS-equipped platforms and environment types. These results were obtained using the 3GPP channel models for the scattering reflection paradigm, i.e.,  (\ref{eq:terres_3GPP}), (\ref{Eq:UAV}) and (\ref{eq:HAPS_LEO_3GPP}). Also, we assume the use of the maximal number of reflectors $N_{\max}$, dry air atmospheric attenuation, and an average tropospheric scintillation of 0.5 dB for the HAPS and LEO systems. 
\textcolor{black}{We notice that the received power of both terrestrial and UAV systems are insensitive to frequency.}
Indeed, the frequency attenuation is successfully addressed through the deployment of a higher number of reflectors, since $N_{\max}$ increases with $f$. In contrast, HAPS and LEO systems performance is affected by frequency, and deeply by the atmospheric attenuation at specific frequency ranges. Nevertheless, some spectrum regions present a stable received power behaviour, such as below 40 GHz, and in the 150-350 GHz and 500-700 GHz bands. The latter can be fully exploited for high capacity communications. Finally, depending on the operating frequency band and environment (dense urban, urban, or rural), one  system may be more suitable than another (e.g., the RSS-equipped HAPS is the most interesting one in rural areas for most of the frequency bands).

\subsection{Data Rate Evaluation and Impact of Platform Location}

For the following simulations, we consider the same assumptions as in Fig. \ref{fig:3GPP_PR_50_dry}. 
Subsequently,  (\ref{eq:terres_3GPP}), (\ref{Eq:UAV}) and (\ref{eq:HAPS_LEO_3GPP}) can be used to evaluate the received power. The related data rate is calculated as 
\begin{equation}
\label{eq:data_rate}
    \mathcal{R} = B_w \log_2\left(1 + \frac{P_r}{P_N}\right),
\end{equation}
where 
$B_w$ denotes the bandwidth, $F$ stands for the noise figure, and $P_N$ refers to the noise power, given by
\begin{equation}
    P_N = K T B_w F,
\end{equation}
where $K= 1.38 \times 10^{-23} \rm{J}.^{\circ}\rm{K}^{-1}$ is the Boltzmann constant and $T$ is the temperature in $^{\circ}$K. 
According to \cite{3gpp2017Technical}, when aerial networks operate in frequency bands $f \geq 6$ GHz, $B_w$ can be up to $800$ MHz in both uplink and downlink, while $F=7$ dB.

\begin{figure}
    \centering
    \includegraphics[width=0.93\linewidth]{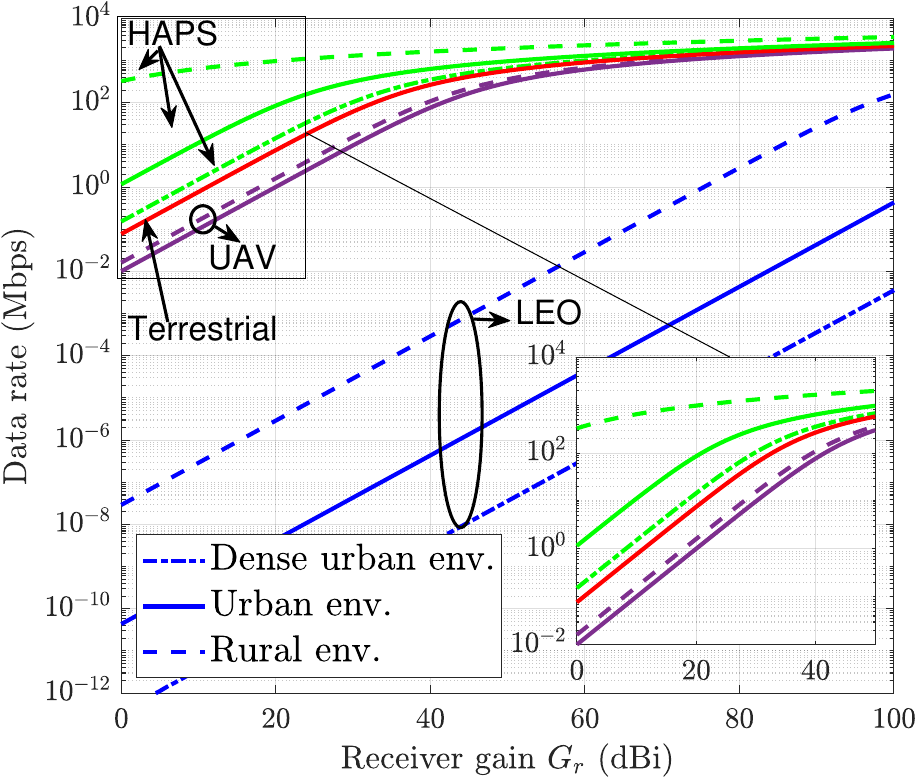}
    \caption{Data rate vs. receiver antenna gain.}
    \label{fig:3GPP_Rate_gain}
\end{figure}

Assuming that each platform uses $N_{\max}$ reflectors and that $B_w=100$ MHz, we depict in Fig. \ref{fig:3GPP_Rate_gain} the resulting data rates as a function of the receiver gain $G_r$. As typical UE has $G_r \leq 5$ dBi, 
\textcolor{black}{both terrestrial and HAPS systems equipped with RSS can}
directly support the downlink communication to users, and this is mainly due to their large RSS areas. However, RSS-equipped \textcolor{black}{UAVs} and LEO \textcolor{black}{satellites} are unable \textcolor{black}{to match this due to the small RSS areas of UAVs and the high propagation loss of LEO satellites.} Alternatively, mounting RSS over UAVs and LEO satellites may \textcolor{black}{support inter-UAV and inter-satellite} communications \cite{tekb2020reconfigurable}. \textcolor{black}{Also, a swarm of RSS-equipped UAVs can be utilized to assist communications cooperatively \cite{alfattani2020aerial,shang2021uav}.}

The platform placement has an important impact on the data rate performance. \textcolor{black}{Unlike LEO satellites, UAVs and HAPS nodes can be placed at fixed positions} above \textcolor{black}{an} intended coverage area. However, due to wind and turbulence, these platforms may drift from \textcolor{black}{their} initial position, thus \textcolor{black}{degrading communication performance.} To assess such an effect, we present in Fig. \ref{fig:placement_Rate} the data rate performance as a function of $\nu=2-\frac{r}{d}$, the normalized horizontal Rx-RSS distance. Moreover, we identify the optimal aerial platform location as calculated by (\ref{Eq:r_cases}) and the one provided by the simulations. Here, we assume that $G_r=0$ dBi, while the remaining parameters are as for Fig. \ref{fig:3GPP_Rate_gain}. 

As discussed \textcolor{black}{above}, the optimal placement \textcolor{black}{of an RSS-equipped aerial platform depends on} the latter's altitude $H_{RSS}$ and coverage radius $d$. For the RSS-equiped UAV, the obtained optimal UAV location $r^*$ that achieves the highest data rate (red dot) agrees with that of  (\ref{Eq:r_cases}) (blue circle), for any environment type. For the RSS-equipped HAPS, we distinguish between two cases, namely for $d=50$ km and $d=10$ km. In both cases, the optimal simulated HAPS locations and those of  (\ref{Eq:r_cases}) agree in the rural environment, but the latter drift away as the environment becomes urbanized. This is mainly due \textcolor{black}{to the effect of the additional shadowing and NLoS links in the 3GPP model} of urban environments, which were ignored in the calculation of  (\ref{Eq:r_cases}). Moreover, this location gap is larger for $d=10$ km due to a higher shadowing impact. For the terrestrial networks, the best location is either being the closest to Tx or Rx, with a preference for Tx. Indeed, since the BS (at altitude 25 m) has a strong LoS link towards the RSS, the received signal is slightly better than  being the closest to Rx (at altitude 1.5 m). 
\begin{figure}
    \centering
    \includegraphics[width=0.98\linewidth]{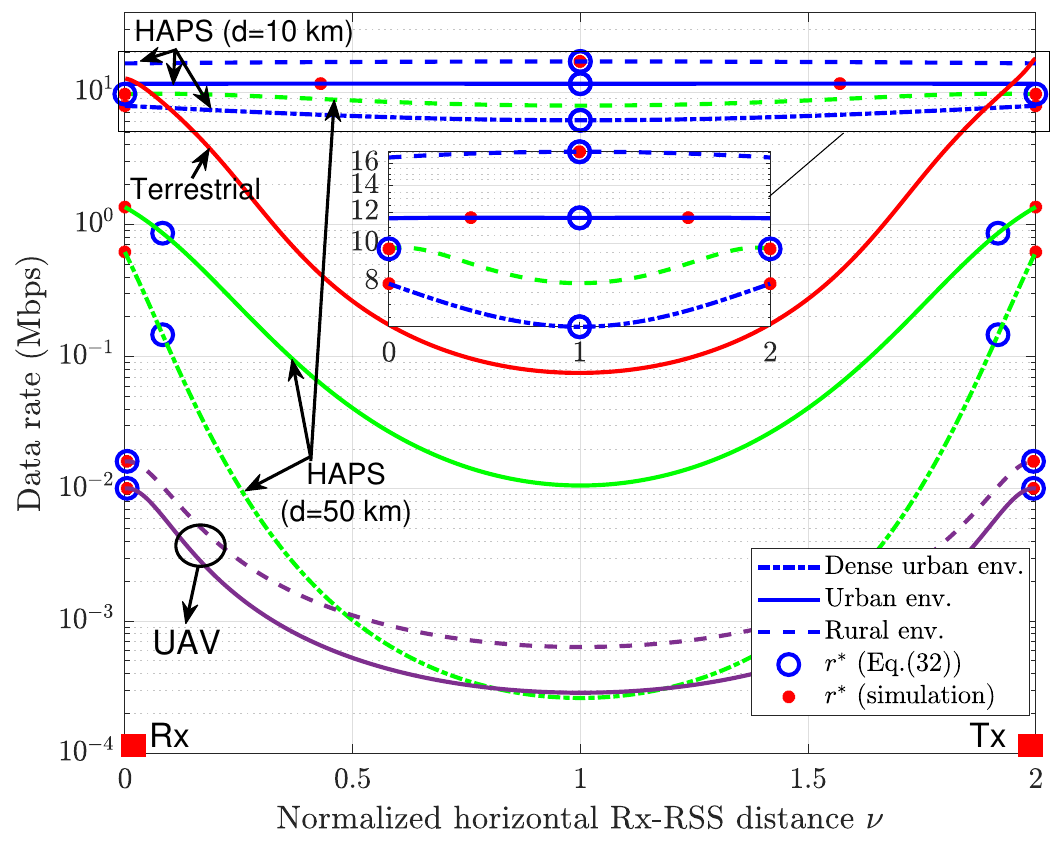}
    \caption{Data rate vs. normalized Rx-RSS distance $\nu$ (different platforms and environment types).}
    \label{fig:placement_Rate}
\end{figure}

When an RSS-equipped HAPS operates in a rural environment, drifting from its initial location would have a small impact on \textcolor{black}{communication performance}. However, the data rate significantly degrades \textcolor{black}{in an urban environment} when the HAPS moves towards the middle of the Tx-Rx segment. \textcolor{black}{A} terrestrial RSS achieves the highest performance at optimal locations. However, \textcolor{black}{due to the inflexibility of its deployment and varying Tx and Rx locations,} it may not perform at its best. In any case, when designed to assist cellular networks, \textcolor{black}{RSS deployed closest to} the BSs would eventually guarantee operating at near-optimal performance.

{\color{black}
\subsection{Outage Probability Analysis}
\textcolor{black}{Links between aerial platforms and terrestrial terminals may be subject to} random variations due to shadowing and blockages. The  instantaneous received power at a terrestrial terminal can be generally written  as
\begin{equation}
    P_r = \overline{P_r} + X,
\end{equation}
where $\overline{P_r}$ is the average received power, defined as 
\begin{equation}
    \overline{P_r} = P_t+G_t+G_r-(\overline{PL_{Tx-RSS}} + \overline{PL_{RSS-Rx}}) + 20\log(N),
\end{equation}
where  $\overline{PL_{Tx-RSS}}$ and  $\overline{PL_{RSS-Rx}}$ denote the average \textcolor{black}{path loss} of links between the RSS-equipped aerial platform and the transmitter and the receiver, respectively. Also, $X$ represents the resulting shadow fading of both links,
modeled as a zero-mean normal distribution with a standard deviation $\sigma_s = \sqrt{\sigma_{Tx-RSS}^2+\sigma_{RSS-Rx}^2}$.  
The  probability density function (pdf) of the received power can be written as
\begin{equation}
    f(P_r) = \frac{1}{\sigma_s \sqrt{2\pi}} \; {\rm exp}\left(- \frac{(P_r-\overline{P_r})^2}{2\sigma_s^2}\right),\; P_r \geq 0.
\end{equation}
Accordingly, the outage probability can be obtained from the cumulative distribution function (cdf) as
\begin{equation}\label{eq:CDF}
    \mathcal{P}_{\textrm{out}}=\mathbb{P}(P_r \leq x) =  \int_{0}^{x} f(P_r) \,dP_r = 1-\frac{1}{2} {\rm erfc}\left(\frac{x-\overline{P_r}}{\sigma_s \sqrt{2}}\right),
\end{equation}
where $x$ reflects the receiver sensitivity and ${\rm erfc}(x)= \frac{2}{\sqrt{\pi}} \int_x^\infty {\rm exp}(-t^2) \,dt$ is the complementary error function.




Using the 3GPP models \textcolor{black}{for the} scattering reflection paradigm (Table \ref{Tab1}) with the characteristics of aerial  platforms (Table \ref{tab:platforms_chr}) and same assumptions as in Fig. \ref{fig:3GPP_PR_50_dry}, we depict in Fig. \ref{fig:prob_pr} the cdf of the received power at user terminals with 0 dB gain, and assisted by RSS mounted on different aerial platforms in various environments. For any aerial platform, as the receiver's sensitivity degrades \textcolor{black}{(i.e., $x$ becomes larger)}, the outage probability increases. This performance degradation is more significant in \textcolor{black}{denser urban} environments due to higher shadowing loss.
Nevertheless, the outage performance gap between rural and urban \textcolor{black}{environments} is more noticeable for \textcolor{black}{a HAPS-assisted communication than for a UAV-assisted one.} 


\begin{figure}
    \centering
    \includegraphics[width=\linewidth]{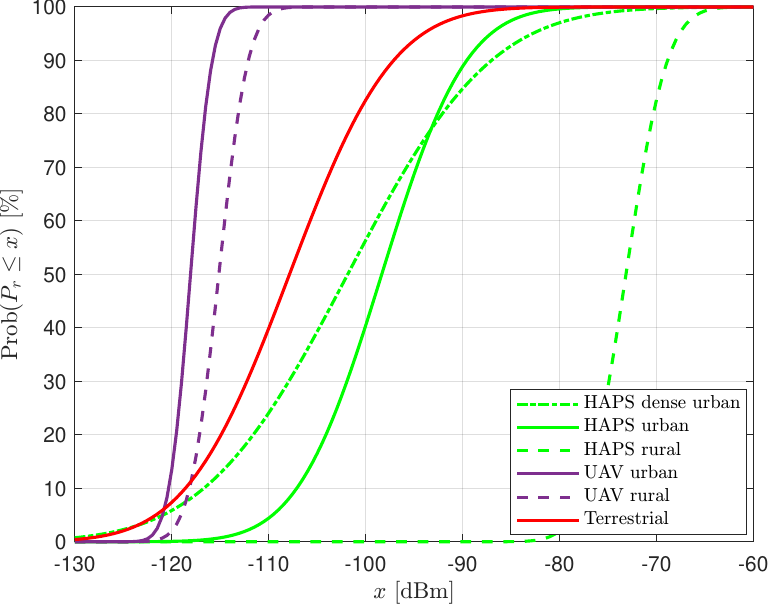}
    \caption{cdf of the received power at ground terminals assisted by RSS-equipped aerial platforms (different environment types).}
    \label{fig:prob_pr}
\end{figure}

The outage probability performance depends on the coverage area of the aerial platform.
To analyze the impact of the coverage area, we illustrate in Fig. \ref{fig:prob_outage_R} the outage performance as a function of the coverage radius. For these simulations, we assume a receiver power sensitivity \textcolor{black}{of} $P_{r}^{\rm th} = -115$ dBm. From Fig. \ref{fig:prob_outage_R}, an outage probability lower than $10 \%$ requires a coverage radius below 0.5 km for a terrestrial RSS and below 1 km for \textcolor{black}{a UAV-assisted communication,} respectively, whereas \textcolor{black}{a HAPS-assisted communication} can support an area \textcolor{black}{of} between 40 km and 80 km in radius ($\approx$ 80 times larger than for a terrestrial RSS). This demonstrates the high potential of deploying RSS on HAPS compared to the alternatives.   

\begin{figure}
    \centering
    \includegraphics[width=\linewidth]{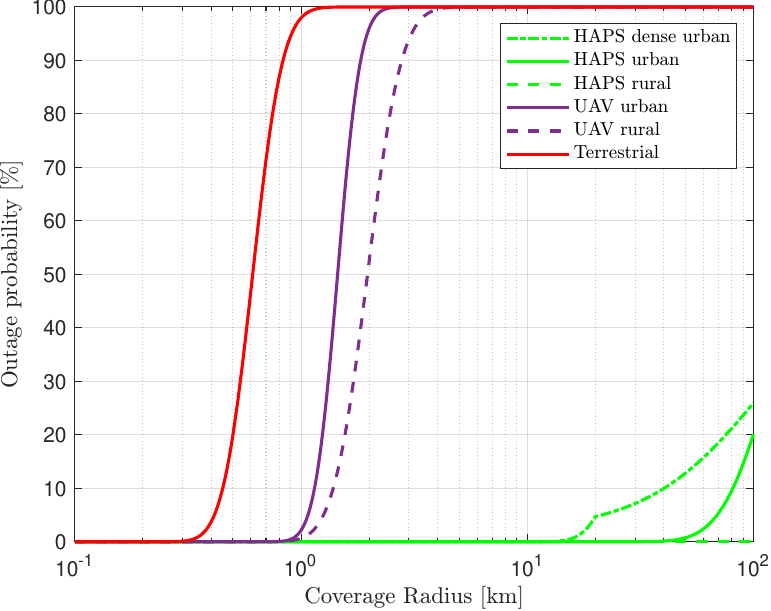}
    \caption{Outage probability vs. coverage radius for RSS-assisted communications (different environment types).}
    \label{fig:prob_outage_R}
\end{figure}

}

\section{Conclusion}\label{sec:conc}
In this paper, we conducted \textcolor{black}{a} link budget analysis for the envisioned RSS-equipped aerial platforms, namely UAVs, HAPS \textcolor{black}{nodes}, and LEO satellites. \textcolor{black}{In a review of the literature,} we identified two reflection paradigms: specular and scattering. For each reflection paradigm, we discussed its realization conditions \textcolor{black}{and then derived} the associated optimal RSS-equipped platform location that maximizes the received power. Numerical results \textcolor{black}{provided a number of insights into the design conditions} of RSS-assisted communications:
\begin{itemize}
    \item The scattering paradigm \textcolor{black}{has gained} more interest in the research community than the specular one due to \textcolor{black}{the} practical accuracy \textcolor{black}{of the former}.
    \item RSS-equipped HAPS presents superior performance in different types of environments, compared to \textcolor{black}{RSS-based} terrestrial and other aerial platforms communications.
    \item The received power performance is limited by the sizes of the RSS \textcolor{black}{area} and \textcolor{black}{number of} reflectors.
    \item When using the maximal number of reflectors \textcolor{black}{for an RSS}, the link budget of the scattering reflection paradigm becomes independent from the carrier frequency.
    \item Supporting ground users with RSS-equipped UAVs and \textcolor{black}{LEO satellites} might not be feasible. Nevertheless, RSS can be used to assist inter-UAV or inter-LEO communications. Also, a swarm of RSS-equipped UAVs can support terrestrial users.
    \item  The best RSS-equipped platform location  depends \textcolor{black}{greatly} on the operating altitude, coverage footprint, and environment type.
    \item {Finally, 
    unlike other RSS-assisted communications, 
    \textcolor{black}{the ones assisted by an RSS-equipped HAPS} sustain the best outage probability/coverage \textcolor{black}{performance} in any type of environment.
    }
    
\end{itemize}


\bibliographystyle{IEEEtran}
\bibliography{IEEEabrv,Final_accepted_version}

\begin{thebibliography}{10}
\providecommand{\url}[1]{#1}
\csname url@samestyle\endcsname
\providecommand{\newblock}{\relax}
\providecommand{\bibinfo}[2]{#2}
\providecommand{\BIBentrySTDinterwordspacing}{\spaceskip=0pt\relax}
\providecommand{\BIBentryALTinterwordstretchfactor}{4}
\providecommand{\BIBentryALTinterwordspacing}{\spaceskip=\fontdimen2\font plus
\BIBentryALTinterwordstretchfactor\fontdimen3\font minus
  \fontdimen4\font\relax}
\providecommand{\BIBforeignlanguage}[2]{{%
\expandafter\ifx\csname l@#1\endcsname\relax
\typeout{** WARNING: IEEEtran.bst: No hyphenation pattern has been}%
\typeout{** loaded for the language `#1'. Using the pattern for}%
\typeout{** the default language instead.}%
\else
\language=\csname l@#1\endcsname
\fi
#2}}
\providecommand{\BIBdecl}{\relax}
\BIBdecl

\bibitem{alfattani2020aerial}
S.~{Alfattani}, W.~{Jaafar}, Y.~{Hmamouche}, H.~{Yanikomeroglu},
  A.~Yonga{\c{c}}oglu, N.~{\DJ}{\`a}o, and P.~{Zhu}, ``Aerial platforms with
  reconfigurable smart surfaces for {5G} and beyond,'' \emph{IEEE Commun.
  Mag.}, vol.~59, no.~1, pp. 96--102, Jan. 2021.

\bibitem{kurt2020vision}
G.~Kurt, M.~G. Khoshkholgh, S.~Alfattani, A.~Ibrahim, T.~S. Darwish, M.~S.
  Alam, H.~Yanikomeroglu, and A.~Yongacoglu, ``A vision and framework for the
  high altitude platform station {(HAPS)} networks of the future,'' \emph{IEEE
  Commun. Surveys. Tuts.}, vol.~23, no.~2, pp. 729--779, Secondquarter 2021.

\bibitem{rinaldi2020non}
F.~Rinaldi, H.-L. Maattanen, J.~Torsner, S.~Pizzi, S.~Andreev, A.~Iera,
  Y.~Koucheryavy, and G.~Araniti, ``Non-terrestrial networks in {5G} \& beyond:
  A survey,'' \emph{IEEE Access}, vol.~8, pp. 165\,178--165\,200, 2020.

\bibitem{kodheli2020satellite}
O.~Kodheli, E.~Lagunas, N.~Maturo, S.~K. Sharma, B.~Shankar, J.~Montoya,
  J.~Duncan, D.~Spano, S.~Chatzinotas, S.~Kisseleff \emph{et~al.}, ``Satellite
  communications in the new space era: A survey and future challenges,''
  \emph{IEEE Commun. Surveys Tuts}, vol.~23, no.~1, pp. 70--109, Firstquarter
  2021.

\bibitem{zeng2019accessing}
Y.~Zeng, Q.~Wu, and R.~Zhang, ``Accessing from the sky: A tutorial on {UAV}
  communications for {5G} and beyond,'' \emph{Proc. IEEE}, vol. 107, no.~12,
  pp. 2327--2375, Dec. 2019.

\bibitem{3gpp2017Technical}
{\it Technical Specification Group Radio Access Network; {S}tudy on New Radio
  {(NR)} to Support Non-Terrestrial Networks}, 3GPP TR 38.811 V15.4.0, Sep.
  2020.

\bibitem{3gpp2017Technical_2}
{\it Technical Specification Group Services and System Aspects; Enhancement for
  Unmanned Aerial Vehicles}, 3GPP TR 22.829 V17.1.0, Sep. 2019.

\bibitem{3gpp2017Technical_3}
{\it Technical Specification Group Services and System Aspects; Unmanned Aerial
  System (UAS) Support in 3GPP}, 3GPP TS 22.125 V17.2.0, Sep. 2020.

\bibitem{starlink}
\BIBentryALTinterwordspacing
``Starlink.'' [Online]. Available: \url{https://www.starlink.com/}
\BIBentrySTDinterwordspacing

\bibitem{thales}
\BIBentryALTinterwordspacing
``Stratobus: Why this stratospheric airship is already being called a ``swiss
  knife'' in the sky.'' [Online]. Available:
  \url{https://www.thalesgroup.com/en/worldwide/space/magazine/stratobus-why-stratospheric-airship-already-being-called-swiss-knife-sky}
\BIBentrySTDinterwordspacing

\bibitem{Nokia}
\BIBentryALTinterwordspacing
``Nokia drone networks.'' [Online]. Available:
  \url{https://dac.nokia.com/applications/nokia-drone-networks/}
\BIBentrySTDinterwordspacing

\bibitem{di2019smart}
M.~Di~Renzo, M.~Debbah, D.-T. Phan-Huy, A.~Zappone, M.-S. Alouini, C.~Yuen,
  V.~Sciancalepore, G.~C. Alexandropoulos, J.~Hoydis, H.~Gacanin \emph{et~al.},
  ``Smart radio environments empowered by reconfigurable {AI} meta-surfaces: An
  idea whose time has come,'' \emph{EURASIP J. Wireless Commun. Netw.}, vol.
  129, no.~1, pp. 1--20, May 2019.

\bibitem{liaskos2018new}
C.~Liaskos, S.~Nie, A.~Tsioliaridou, A.~Pitsillides, S.~Ioannidis, and
  I.~Akyildiz, ``A new wireless communication paradigm through
  software-controlled metasurfaces,'' \emph{IEEE Commun. Mag.}, vol.~56, no.~9,
  pp. 162--169, Sep. 2018.

\bibitem{Basar2019}
E.~Basar, M.~{Di Renzo}, J.~{De Rosny}, M.~Debbah, M.-S. Alouini, and R.~Zhang,
  ``{Wireless communications through reconfigurable intelligent surfaces},''
  \emph{IEEE Access}, vol.~7, pp. 116\,753--116\,773, 2019.

\bibitem{Wu2019}
Q.~{Wu} and R.~{Zhang}, ``Intelligent reflecting surface enhanced wireless
  network via joint active and passive beamforming,'' \emph{IEEE Trans.
  Wireless Commun.}, vol.~18, no.~11, pp. 5394--5409, Nov. 2019.

\bibitem{Tan2018}
X.~Tan, Z.~Sun, D.~Koutsonikolas, and J.~M. Jornet, ``{Enabling Indoor Mobile
  Millimeter-wave Networks Based on Smart Reflect-arrays},'' in \emph{Proc.
  IEEE Conf. Comput. Commun. (INFOCOM)}, Honolulu, HI, USA, 2018.

\bibitem{wu2020intelligent}
Q.~{Wu}, S.~{Zhang}, B.~{Zheng}, C.~{You}, and R.~{Zhang}, ``Intelligent
  reflecting surface aided wireless communications: A tutorial,'' \emph{IEEE
  Trans. Commun.}, vol.~69, no.~5, pp. 3313--3351, May 2021.

\bibitem{gong2020toward}
S.~Gong, X.~Lu, D.~T. Hoang, D.~Niyato, L.~Shu, D.~I. Kim, and Y.-C. Liang,
  ``Toward smart wireless communications via intelligent reflecting surfaces: A
  contemporary survey,'' \emph{IEEE Commun. Surveys Tuts}, vol.~22, no.~4, pp.
  2283--2314, Fourthquarter 2020.

\bibitem{lu2020enabling}
H.~Lu, Y.~Zeng, S.~Jin, and R.~Zhang, ``Enabling panoramic full-angle
  reflection via aerial intelligent reflecting surface,'' in \emph{proc. 2020
  IEEE Int. Conf. Commun. Wksh. (ICC Wksh.)}, Dublin, Ireland, Jun. 2020, pp.
  1--6.

\bibitem{abdalla2020uavs}
A.~S. Abdalla, T.~F. Rahman, and V.~Marojevic, ``{UAVs} with reconfigurable
  intelligent surfaces: Applications, challenges, and opportunities,''
  \emph{arXiv preprint, [Online]. Available: https://arxiv.org/abs/2012.04775},
  Dec. 2020.

\bibitem{shang2021uav}
B.~Shang, R.~Shafin, and L.~Liu, ``{UAV} swarm-enabled aerial reconfigurable
  intelligent surface,'' \emph{arXiv preprint, [online]. Available:
  https://arxiv.org/abs/2103.06361}, 2021.

\bibitem{samir2020optimizing}
M.~Samir, M.~Elhattab, C.~Assi, S.~Sharafeddine, and A.~Ghrayeb, ``Optimizing
  age of information through aerial reconfigurable intelligent surfaces: A deep
  reinforcement learning approach,'' \emph{IEEE Trans. Veh. Technol.}, vol.~70,
  no.~4, pp. 3978--3983, Apr. 2021.

\bibitem{tekb2020reconfigurable}
K.~{Tekbiyik}, G.~K. {Kurt}, A.~{Ekti}, A.~{Görcin}, and H.~Yanikomeroglu,
  ``Reconfigurable intelligent surface empowered terahertz communication for
  {LEO} satellite networks,'' \emph{arXiv preprint, [Online]. Available:
  https://arxiv.org/abs/2007.04281}, Dec. 2020.

\bibitem{Tang2019}
W.~Tang, M.~Z. Chen, X.~Chen, J.~Y. Dai, Y.~Han, M.~{Di Renzo}, Y.~Zeng,
  S.~Jin, Q.~Cheng, and T.~J. Cui, ``Wireless communications with
  reconfigurable intelligent surface: path loss modeling and experimental
  measurement,'' \emph{IEEE Trans. Wireless Commun.}, vol.~20, no.~1, pp.
  421--439, Jan. 2021.

\bibitem{yildirim2020modeling}
I.~Yildirim, A.~Uyrus, and E.~Basar, ``Modeling and analysis of reconfigurable
  intelligent surfaces for indoor and outdoor applications in future wireless
  networks,'' \emph{IEEE Trans. Commun.}, vol.~69, no.~2, pp. 1290--1301, Feb.
  2021.

\bibitem{Ozdogan2019}
O.~Özdogan, E.~Björnson, and E.~G. Larsson, ``Intelligent reflecting
  surfaces: physics, propagation, and pathloss modeling,'' \emph{IEEE Wireless
  Commun. Lett.}, vol.~9, no.~5, pp. 581--585, May 2020.

\bibitem{nadeem2019intelligent}
Q.-U.-A. Nadeem, A.~Kammoun, A.~Chaaban, M.~Debbah, and M.-S. Alouini,
  ``Intelligent reflecting surface assisted wireless communication: Modeling
  and channel estimation,'' \emph{arXiv preprint, [Online]. Available:
  https://arxiv.org/abs/1906.02360}, Dec. 2019.

\bibitem{Ntontin2019a}
M.~{Di Renzo}, K.~Ntontin, J.~Song, F.~H. Danufane, X.~Qian, F.~Lazarakis,
  J.~de~Rosny, D.~T. Phan-Huy, O.~Simeone, R.~Zhang, M.~Debbah, G.~Lerosey,
  M.~Fink, S.~Tretyakov, and S.~Shamai, ``Reconfigurable intelligent surfaces
  vs. relaying: Differences, similarities, and performance comparison,''
  \emph{IEEE Open J. Commun. Soc.}, vol.~1, p. 798–807, 2020.

\bibitem{basar2020simris}
E.~Basar and I.~Yildirim, ``Simris channel simulator for reconfigurable
  intelligent surface-empowered communication systems,'' in \emph{proc. 2020
  IEEE Latin-American Conf. Commun. (LATINCOM)}, Santo Domingo, Dominican
  Republic, Nov. 2020, pp. 1--6.

\bibitem{ellingson2019path}
S.~W. Ellingson, ``Path loss in reconfigurable intelligent surface-enabled
  channels,'' \emph{arXiv preprint, [Online]. Available:
  https://arxiv.org/abs/1912.06759}, Dec. 2019.

\bibitem{Gab2020}
J.~C.~B. {Garcia}, A.~{Sibille}, and M.~{Kamoun}, ``Reconfigurable intelligent
  surfaces: Bridging the gap between scattering and reflection,'' \emph{IEEE J.
  Sel. Areas Commun.}, vol.~38, no.~11, pp. 2538--2547, Nov. 2020.

\bibitem{bjornson2021reconfigurable}
E.~Björnson, H.~Wymeersch, B.~Matthiesen, P.~Popovski, L.~Sanguinetti, and
  E.~de~Carvalho, ``Reconfigurable intelligent surfaces: A signal processing
  perspective with wireless applications,'' \emph{arXiv preprint, [online].
  Available: https://arxiv.org/abs/2102.00742}, 2021.

\bibitem{Hum2014}
S.~V. Hum and J.~Perruisseau-Carrier, ``{Reconfigurable reflectarrays and array
  lenses for dynamic antenna beam Control: A review},'' \emph{IEEE Trans.
  Antennas Propag.}, vol.~62, no.~1, pp. 183--198, Jan. 2014.

\bibitem{wu2019beamforming}
Q.~Wu and R.~Zhang, ``Beamforming optimization for wireless network aided by
  intelligent reflecting surface with discrete phase shifts,'' \emph{IEEE
  Trans. Commun.}, vol.~68, no.~3, pp. 1838--1851, Mar. 2020.

\bibitem{huang2018energy}
C.~Huang, G.~C. Alexandropoulos, A.~Zappone, M.~Debbah, and C.~Yuen, ``Energy
  efficient multi-user miso communication using low resolution large
  intelligent surfaces,'' in \emph{IEEE Globecom Workshops (GC Wkshps)}.\hskip
  1em plus 0.5em minus 0.4em\relax IEEE, 2018, pp. 1--6.

\bibitem{rappaport1996wireless}
T.~S. Rappaport, \emph{Wireless Communications: Principles and Practice}.\hskip
  1em plus 0.5em minus 0.4em\relax Prentice Hall PTR, 1996.

\bibitem{myths}
E.~{Björnson}, O.~{Özdogan}, and E.~G. {Larsson}, ``Reconfigurable
  intelligent surfaces: Three myths and two critical questions,'' \emph{IEEE
  Commun. Mag.}, vol.~58, no.~12, pp. 90--96, Dec. 2020.

\bibitem{3gpp38901study}
{\it Study on Channel Model for Frequencies from 0.5 to 100 {GHz}}, 3GPP TR
  38.901 V14.3.0, Jan. 2018.

\bibitem{itu2020}
{\it Prediction of Building Entry Loss}, ITU-R P.2109-1 P Series, Aug. 2019.

\bibitem{3gpp2017enhanced}
{\it Study on Enhanced {LTE} Support for Aerial Vehicles}, 3GPP TR 36.777
  V1.0.0, Dec. 2017.

\bibitem{alzenad20173}
M.~Alzenad, A.~El-Keyi, F.~Lagum, and H.~Yanikomeroglu, ``{3-D} placement of an
  unmanned aerial vehicle base station ({UAV-BS}) for energy-efficient maximal
  coverage,'' \emph{IEEE Wireless Commun. Lett.}, vol.~6, no.~4, pp. 434--437,
  Aug. 2017.

\bibitem{Mozaffari2019}
M.~Mozaffari, W.~Saad, M.~Bennis, Y.-H. Nam, and M.~Debbah, ``{A tutorial on
  {UAV}s for wireless networks: Applications, challenges, and open problems},''
  \emph{IEEE Commun. Surveys. Tuts.}, vol.~21, no.~3, pp. 2334--2360,
  thirdquarter 2019.

\bibitem{sahabul2021}
M.~S. {Alam}, G.~K. {Kurt}, H.~{Yanikomeroglu}, P.~{Zhu}, and N.~D. {Đào},
  ``High altitude platform station based super macro base station
  constellations,'' \emph{IEEE Commun. Mag.}, vol.~59, no.~1, pp. 103--109,
  Jan. 2021.

\bibitem{ITU_F1500}
\emph{{\it Preferred Characteristics of Systems in the Fixed Service Using High
  Altitude Platforms Operating in the Bands 47.2-47.5 GHz and 47.9-48.2 GHz}},
  ITU-R F.1500, May 2000.

\bibitem{fossa1998overview}
C.~E. Fossa, R.~A. Raines, G.~H. Gunsch, and M.~A. Temple, ``An overview of the
  {I}ridium {(R)} low earth orbit {(LEO)} satellite system,'' in \emph{Proc.
  IEEE Nat. Aerospace Electron. Conf. (NAECON)}, Dayton, OH, USA, 1998, pp.
  152--159.

\bibitem{cakaj2016coverage}
S.~Cakaj, ``The coverage belt for low earth orbiting satellites,'' in
  \emph{Proc. 39th Int. Convent. Info. Commun. Technol., Electron.
  Microelectron. (ICICTEM)}, Opatija, Croatia, May 2016, pp. 554--557.

\bibitem{su2019broadband}
Y.~Su, Y.~Liu, Y.~Zhou, J.~Yuan, H.~Cao, and J.~Shi, ``Broadband {LEO}
  satellite communications: Architectures and key technologies,'' \emph{IEEE
  Wireless Commun.}, vol.~26, no.~2, pp. 55--61, Apr. 2019.

\bibitem{LEO_sat_iridium}
\BIBentryALTinterwordspacing
``Hosting payloads on a communication constellation.'' [Online]. Available:
  \url{https://directory.eoportal.org/web/eoportal/satellite-missions/i/iridium-next}
\BIBentrySTDinterwordspacing

\bibitem{LEO_sat_est}
\BIBentryALTinterwordspacing
``Starlink satellite dimension estimates.'' [Online]. Available:
  \url{https://lilibots.blogspot.com/2020/04/starlink-satellite-dimension-estimates.html}
\BIBentrySTDinterwordspacing

\bibitem{sector2013recommendation}
{\it Attenuation by Atmospheric Gases}, ITU-R P.676-10, Sep. 2013.

\bibitem{itu1999p}
{\it Reference Standard Atmospheres}, ITU-R P.835-6 P Series, Dec. 2017.

\bibitem{Ionospheric2019ITU}
{\it Ionospheric Propagation Data and Prediction Methods Required for the
  Design of Satellite Networks and Systems}, ITU-R P.531-14 P Series, Aug.
  2019.

\bibitem{series2015propagation}
{\it Propagation Data and Prediction Methods Required for the Design of
  Earth-space Telecommunication Systems}, ITU-R P.618-13 P Series, Dec. 2017.

\end{thebibliography}

\end{document}